\documentclass[a4paper,UKenglish]{lipics-v2021}

\title{Greybox Learning of Languages Recognizable by Event-Recording Automata}
\author{Anirban Majumdar}{Université Libre de Bruxelles, Belgium}{anirban.majumdar@ulb.be}{https://orcid.org/0000-0003-4793-1892}{funded by a grant from Fondation ULB}
\author{Sayan Mukherjee}{Université Libre de Bruxelles, Belgium}{sayan.mukherjee@ulb.be}{https://orcid.org/0000-0001-6473-3172}{funded by a grant from Fondation ULB}
\author{Jean-François Raskin}{Université Libre de Bruxelles, Belgium}{jean-francois.raskin@ulb.be}{https://orcid.org/0000-0002-3673-1097}{receives support from the Fondation ULB}

\authorrunning{Anirban Majumdar, Sayan Mukherjee, and Jean-François Raskin}
\supplement{The source codes related to our algorithm can be found as follows:}
\supplementdetails[subcategory={Source Code}]{Software}{https://github.com/mukherjee-sayan/ERA-greybox-learn}
\supplementdetails[subcategory={}]{ATVA 2024 Artifact Evaluation approved artifact}{https://doi.org/10.5281/zenodo.12684215}

\usepackage[normalem]{ulem}
\usepackage{amsmath, amsfonts, amssymb}
\usepackage{thmtools}
\usepackage{graphicx}
\usepackage{multirow}
\usepackage{xcolor}
\usepackage{todonotes}
\usepackage{subcaption}
\usepackage{pifont}
\usepackage{wrapfig,lipsum,booktabs}

\usepackage{tikz,pgf}
\usetikzlibrary{arrows,automata,shapes,calc,positioning,intersections,backgrounds}

\newcommand{\sem}[1]{[ \! [ #1 ] \! ]}

\newcommand{\word}{\mathsf w}

\newcommand{\lrw}[1]{L^{\mathsf{rw}}(#1)}
\newcommand{\lsrw}[1]{L^{\mathsf {s \cdot rw}}(#1)}
\newcommand{\ltw}[1]{L^{\mathsf{tw}}(#1)}

\newcommand{\tw}{\mathsf{tw}}
\newcommand{\cw}{\mathsf{cw}}
\newcommand{\rw}{\mathsf{rw}}

\newcommand{\sw}{\mathsf{sw}}
\newcommand{\regL}{{\sf RegL}(\Sigma,K)}

\newcommand{\tlsep}{\textsf{tLSep}}

\newcommand{\sevents}{\Sigma \times {\sf Reg}(\Sigma,K)}

\newcommand{\SC}{\mathsf{C}} 
\newcommand{\RW}[1]{\mathsf{RW}_K(#1)}

\newcommand{\ie}{\emph{i.e.}}

\newcommand{\tick}{\ding{51}}
\newcommand{\cross}{\ding{55}}

\nolinenumbers

\begin{document}
\hideLIPIcs
\maketitle

\begin{abstract}
In this paper, we revisit the active learning of timed languages recognizable by event-recording automata. Our framework employs a method known as greybox learning, which enables the learning of event-recording automata with a minimal number of control states. This approach avoids learning the region automaton associated with the language, contrasting with existing methods. We have implemented our greybox learning algorithm with various heuristics to maintain low computational complexity. The efficacy of our approach is demonstrated through several examples.
\end{abstract}

\section{Introduction}
\label{sec:intro}

Formal modeling of complex systems is essential in fields such as requirement engineering and computer-aided verification. However, the task of writing formal models is both laborious and susceptible to errors. This process can be facilitated by learning algorithms~\cite{ACM-Vaandrager}. For example, if a system can be represented as a finite state automaton, and thus if the system's set of behaviors constitutes a regular language, there exist thoroughly understood passive and active learning algorithms capable of learning the minimal DFA of the underlying language\footnote{Minimal DFAs are a canonical representation of regular languages and are closely tied to the Myhill-Nerode equivalence relation}. We call this problem the specification discovery problem in the sequel. In the realm of active learning, the $L^*$ algorithm~\cite{Angluin87}, pioneered by Angluin, is a celebrated approach that employs a protocol involving membership and equivalence queries, and enabling an efficient learning process for regular languages.

In this paper, we introduce an active learning framework designed to acquire timed languages that can be recognized by {\em event-recording automata}
(ERA, for short)\cite{AFH99}.
The class of ERA enjoys several nice properties over the more general class of Timed Automata (TA)~\cite{AD94}. First, ERA are determinizable and hence complementable, unlike TA, making it particularly well suited for verification as inclusion checking is decidable for this class. Second, in ERA, clocks are explicitly associated to events and implicitly reset whenever its associated event occurs. This property leads to automata that are usually much easier to interpret than classical TA where clocks are not tied to events. This makes the class of ERA a very interesting candidate for learning, when interpretability is of concern.

Adapting the $L^*$ algorithm for ERA is feasible by employing a canonical form known as the {\em simple deterministic event-recording automata} (SDERA). These SDERAs, introduced in~\cite{GJL10} specifically for this application, are essentially the underlying region automaton of the timed language. Although the region automaton fulfills all the necessary criteria for the $L^*$ algorithm, its size is generally excessive. Typically, the region automaton is significantly larger than a deterministic ERA (DERA, for short) with the minimum number of locations for a given timed language (and it is known that it can be \emph{exponentially} larger). This discrepancy arises because the region automaton explicitly incorporates the semantics of clocks in its states, whereas in a DERA, the semantics of clocks are implicitly utilized. The difficulty to work directly with DERA is related to the fact that there is no natural canonical form for DERA and usually, there are several DERA with minimum number of states that recognize the same timed language. To overcome this difficulty, we propose in this paper an innovative application of the {\em greybox learning} framework, see for example~\cite{DBLP:conf/nfm/0002R16}, and demonstrate that it is possible to learn an automaton whose number of states matches that of a minimum DERA for the underlying timed language.

In our approach, rather than learning the region automaton of the timed language, we establish an active learning setup in which the semantics of the regions are already {\em known} to the learning algorithm, thus eliminating the need for the algorithm to learn the semantics of regions.
 Our greybox learning algorithm is underpinned by the following concept:
\begin{equation}
\label{eq:intersection}
\RW{L}=L(C) \cap {\sf RegL}(\Sigma,K)
\end{equation}
\noindent
Here, $\RW{L}$ is a set of \emph{consistent region words} that exactly represents the target timed language $L$.
Meanwhile, ${\sf RegL}(\Sigma,K)$ denotes the regular language of consistent region words across the alphabet $\Sigma$ and a maximal constant $K$. This language is precisely defined once $\Sigma$ and $K$ are fixed and remains independent of the timed language to be learned; therefore, it does not require learning. Essentially, deciding if a region word belongs to $\regL$ or not reduces to solving the following “consistency problem”: given a  region word, are there concrete timed words \emph{consistent} with this (symbolic) region word? The answer to this question does \emph{not} depend on the timed language to be learned but \emph{only} on the semantics of clocks. 
This reduces the consistency problem to the problem of solving a linear program, which is known to be solvable efficiently.

The target of our greybox learning algorithm is a DERA $C$ recognizing the timed language $L$. 
We will demonstrate that if a DERA with $n$ states can accept the timed language to be learned, then 
the DERA $C$ returned by our algorithm has at most $n$ many states. 
Here is an informal summary of our learning protocol: within the context of the $L^*$ vocabulary, we assume that the {\tt Teacher} is capable of responding to membership queries of \emph{consistent}
region words. Additionally, the {\tt Teacher} handles \emph{inclusion} queries, unlike equivalence queries in $L^*$ by determining whether a timed language is contained in another.
If a discrepancy exists, {\tt Teacher} provides a counter-example in the form of a region word. 

This learning framework can be applied to specification discovery problem presented in the motivations above: we assume that the requirements engineer has a specific timed language in mind for which he wants to construct a DERA, but we assume that he is not able to formally write this DERA on his own. However, the requirement engineer is capable of answering membership and inclusion queries related to this language. Since the engineer cannot directly draft a DERA that represents the timed language, it is crucial that the queries posed are {\em semantic}, focusing on the language itself rather than on a specific automaton representing the language. This distinction ensures that the questions are appropriate and manageable given the assumptions about the engineer's capabilities. This is in contrast with~\cite{LADSL11}, where queries are syntactic and answered with a specific DERA known by the {\tt Teacher}. We will give more details about this comparison later in Section~\ref{sec:learning-algo}.

\subparagraph*{Technical contribution.} First, we define a greybox learning algorithm to derive a DERA with the {\em minimal number} of states that recognizes a target timed language. To achieve the learning of a minimal automaton, we must solve an NP-hard optimization problem: calculating a minimal finite state automaton that separates regular languages of positive and negative examples derived from membership and equivalence queries~\cite{DBLP:journals/iandc/Gold78}. For that, we adapt to the timed setting the approach of Chen et al. in~\cite{CFCTW09}.
A similar framework has also been studied using SMT-based techniques by Moeller et al. in~\cite{MWSKF023}. Second, we explore methods to circumvent this NP-hard optimization problem when the requirement for minimality is loosened. We introduce a heuristic that computes a candidate automaton in polynomial time using a greedy algorithm. This algorithm strives to maintain a small automaton size but does not guarantee the discovery of the absolute minimal solution, yet 
guaranteed to terminate with a DERA that accepts the input timed language. We have implemented this heuristic and demonstrate with our prototype that, in practice, it yields compact automata.

\subparagraph*{Related works.}
Grinchtein et al. in~\cite{GJL10} also present an active learning framework for ERA-recognizable languages. However, their learning algorithm generates simple DERAs -- these are a canonical form of ERA recognizable languages, closely aligned with the underlying region automaton. In a simple DERA, transitions are labeled with region constraints, and notably, every path leading to an accepting state is satisfiable. This means the sequence of regions is consistent (as in the region automaton), and any two path prefixes reaching the same state are region equivalent. Typically, this results in a prohibitively large number of states in the automaton that maintains explicitly the semantics of regions. In contrast, our approach, while also labeling automaton edges with region constraints, permits the merging of prefixes that are not region equivalent. Enabled by equation \ref{eq:intersection}, this allows for a more compact representation of the ERA languages (we generate DERA with minimal number of states). While our algorithm has been implemented, there is, to the best of our knowledge, no implementation of Grinchtein's algorithm, certainly due to  high complexity of their approach.

Lin et al. in~\cite{LADSL11} propose another active learning style algorithm for inferring ERA, but with zones (coarser constraints than regions) as guards on the transitions. This allows them to infer smaller automata than~\cite{GJL10}.
However, it assumes that {\tt Teacher} is familiar with a {\em specific} DERA for the underlying language and is capable of answering syntactic queries about this DERA. 
In fact, their learning algorithm is recognizing the regular language of words over the alphabet, composed of events and zones of the specific DERA known to {\tt Teacher}. Unlike this approach, our learning algorithm does not presuppose a known and fixed DERA but only relies on queries that are semantic rather than syntactic. Although the work of~\cite{LADSL11} may seem superficially similar to ours, we dedicate in Section~\ref{sec:learning-algo} an entire paragraph to the formal and thorough comparison between our two different approaches.

Recently, Masaki Waga developed an algorithm~\cite{Waga23} for learning the class of Deterministic Timed Automata (DTA), which includes ERA recognizable languages. Waga's method resets a clock on every transition and utilizes region-based guards, often producing complex automata. Consequently, this algorithm does not guarantee the minimality of the resulting automaton, but merely ensures it accepts the specified timed language. In contrast, our approach guarantees learning of automata with the minimal number of states, featuring implicit clock resets that enhance readability. 
In the experimental section, we detail a performance comparison of our tool against their tool, {\sf LearnTA}, across various examples. Further, in~\cite{DTA-hscc}, the authors extend Waga's approach to reduce the number of clocks in learned automata, though they also do not achieve minimality. Their method also requires a significantly higher number of queries to implement these improvements. 

Several other works have proposed learning methods for different models of timed languages, such as, One-clock Timed Automata~\cite{octa}, Real-time Automata~\cite{AWZZZ21}, Mealy-machines with timers~\cite{VB021,BGPSV24}. There have also been approaches other than active learning for learning such models, for example, using Genetic programming~\cite{TALL19} and passive learning~\cite{VWW09,CLRT22}.

\section{Preliminaries and notations}
\label{sec:prelims}
\subparagraph*{DFA and 3DFA.}
Let $\Sigma=\{\sigma_1,\sigma_2,\dots,\sigma_k\}$ be a finite {\em alphabet}. A deterministic finite state automaton (DFA) over $\Sigma$ is a tuple $A=(Q,q_{\sf init},\Sigma,\delta,F)$ where $Q$ is a finite set of states, $q_{\sf init} \in Q$ is the initial state, $\delta : Q \times \Sigma \rightarrow Q$ is a partial transition function, and $F \subseteq Q$ is the set of final (accepting) states. When $\delta$ is a total function, we say that $A$ is \emph{total}. A DFA $A=(Q,q_{\sf init},\Sigma,\delta,F)$ defines a regular language, noted $L(A)$, which is a subset of $\Sigma^*$ and which contains the set of words $w=\sigma_1 \sigma_2 \dots \sigma_n$ such that there exists an {\em accepting run} of $A$ over~$w$, \emph{i.e.}, there exists a sequence of states $q_0 q_1 \dots q_{n}$ such that $q_0=q_{\sf init}$, and for all $i$, $1 \leq i \leq n$, $\delta(q_{i-1},\sigma_{i})=q_{i}$, and $q_{n} \in F$. We then say that the word $w$ is \emph{accepted} by $A$. If such an accepting run does not exist, we say that $A$ \emph{rejects} the word $w$.

A 3DFA $D$ over $\Sigma$ is a tuple $(Q,q_{\sf init},\Sigma,\delta,{\sf A},{\sf R},{\sf E})$ where $Q$ is a finite set of states, $q_{\sf init} \in Q$ is the initial state, $\delta : Q \times \Sigma \rightarrow Q$ is a total transition function, 
and ${\sf A} \uplus {\sf R} \uplus {\sf E}$ forms a partition of $Q$ into a set of accepting states ${\sf A}$, a set of rejection states ${\sf R}$, and a set of `don't care' states~${\sf E}$. The 3DFA $D$ defines the function $D : \Sigma^* \rightarrow \{0,1,?\}$ as follows: for $w=\sigma_0 \sigma_1 \dots \sigma_n$ let $q_0 q_1 \dots q_{n+1}$ be the unique run of $D$ on $w$.  If 
$q_{n+1} \in {\sf A}$ (resp.~$q_{n+1} \in {\sf R}$), then $D(w)=1$ (resp.~$D(w)=0$) and we say that $w$ is \emph{strongly accepted} (resp.~\emph{strongly rejected}) by $D$, and  if $q_{n+1} \in {\sf E}$, then $D(w)={?}$. 
We interpret ``$?$'' as {\em either way}, \emph{i.e.}, so the word can be either accepted or rejected. We say that a regular language $L$ is consistent with a 3DFA $D$, noted
$L \models D$, if for all $w \in \Sigma^*$, if $D(w)=0$ then $w \not\in L$, and if $D(w)=1$ then $w \in L$. 
Accordingly, we define $D^-$ as the DFA obtained from $D$ where the partition ${\sf A}\uplus{\sf R}\uplus{\sf E}$ is replaced by $F={\sf R}$, clearly $L(D^-)=\{ w \in \Sigma^* \mid D(w)=0 \}$, it is thus the set of words that are strongly rejected by $D$. Similarly, we define $D^+$ as the DFA obtained from $D$ where the partition ${\sf A}\uplus{\sf R}\uplus{\sf E}$ is replaced by $F={\sf A}$, then $L(D^+)=\{ w \in \Sigma^* \mid D(w)=1 \}$, it is thus the set of words that are strongly accepted by $D$. Now, it is easy to see that $L \models D$ iff $L(D^+) \subseteq L$ and $L(D^-)\subseteq \overline{L}$. 

\subparagraph*{Timed words and timed languages.}
A \emph{timed word} over an alphabet $\Sigma$ is a finite sequence $(\sigma_1,t_1)(\sigma_2,t_2)\dots (\sigma_n,t_n)$ where each $\sigma_i \in \Sigma$ and $t_i \in \mathbb{R}_{\geq 0}$, for all $1 \leq i \leq n$, and for all $1 \leq i < j \leq n$, $t_i \leq t_j$ (time is monotonically increasing). We use $\mathsf{TW}(\Sigma)$ to denote the set of all timed words over the alphabet $\Sigma$.

A \emph{timed language} is a (possibly infinite) set of timed words.
In existing literature, a timed language is called \emph{timed regular} when there exists a timed automaton `recognizing' the language. Timed Automata (TA)~\cite{AD94} extend deterministic finite-state automata with \emph{clocks}. In what follows, we will 
use a subclass of TA, where clocks have a pre-assigned semantics and are not reset arbitrarily. This class is known as the class of Event-recording Automata (ERA)~\cite{AFH99}.
We now introduce the necessary vocabulary and notations for their definition.

\subparagraph*{Constraints.} 
A \emph{clock} is a non-negative real valued variable, that is, a variable ranging over $\mathbb{R}_{\geq 0}$.
Let $K$ be a positive integer.
An \emph{atomic $K$-constraint} over a clock $x$, is an expression of the form $x = c$, $x \in (c,d)$ or $x > K$ where $c,d \in \mathbb{N} \cap [0,K]$.
A \emph{$K$-constraint} over a set of clocks $X$ is a conjunction of atomic $K$-constraints over clocks in $X$.
Note that, $K$-constraints are closely related to the notion of \emph{zones}~\cite{DBLP:conf/avmfss/Dill89} as in the literature of TA. 
However, zones also allow difference between two clocks, that are not allowed in $K$-constraints.

An \emph{elementary $K$-constraint} over a clock $x$, is an atomic $K$-constraint where the intervals are restricted to unit intervals; more formally, it is an expression of the form $x = c$, $x \in (d,d+1)$ or $x > K$ where $c,d,d+1 \in \mathbb{N} \cap [0,K]$. 
A \emph{simple $K$-constraint} over $X$ is a conjunction of elementary $K$-constraints over clocks in $X$, 
where each variable $x \in X$ appears exactly in one conjunct. 
The definition of simple constraints also appear in earlier works, for example~\cite{GJL10}.
One can again note that, simple constraints closely relate to the classical notion of \emph{regions}~\cite{AD94}.
However, again, regions consider difference between two clocks as well, which simple constraints do not. Interestingly, we do not need the granularity that regions offer (by comparing clocks) for our purpose, because a sequence of simple constraints induce a unique region in the classical sense (see Lemma~\ref{lem:simple-region}).

\subparagraph*{Satisfaction of constraints.}
A \emph{valuation} for a set of clocks $X$ is a function ${v \colon X \rightarrow \mathbb{R}_{\geq 0}}$.
A valuation $v$ for $X$ satisfies an atomic $K$-constraint $\psi$, written as $v \models \psi$, if: $v(x) = c$ when $\psi := x = c$, $v(x) \in (c,d)$ when $\psi := x \in (c,d)$ and $v(x) > K$ when $\psi := x > K$.

A valuation $v$ satisfies a $K$-constraint over $X$ if $v$ satisfies all its conjuncts.
Given a $K$-constraint $\psi$, we denote by $\sem{\psi}$ the set of all valuations $v$ such that $v \models \psi$. We say that a simple constraint $r$ satisfies an elementary constraint $\psi$, noted $r \models \psi$, if $\sem{r} \subseteq \sem{\psi}$. It is easy to verify that for any simple constraint $r$ and any elementary constraint $\psi$, either $\sem{r} \cap \sem{\psi}=\emptyset$ or $\sem{r} \subseteq \sem{\psi}$.

Let $\Sigma=\{\sigma_1,\sigma_2,\dots,\sigma_k\}$ be a finite \emph{alphabet}. The set of \emph{event-recording clocks} associated to $\Sigma$ is denoted by $X_{\Sigma}=\{ x_{\sigma} \mid \sigma \in \Sigma\}$. We denote by $\SC(\Sigma, K)$ the set of all $K$-constraints over the set of clocks $X_{\Sigma}$ and we use $\mathsf{Reg}(\Sigma,K)$ to denote the set of all simple $K$-constraints over the clocks in $X_{\Sigma}$. For simple $K$-constraints we adopt the notation $\mathsf{Reg}$, since, as remarked earlier, these kinds of constraints relate closely to the well-known notion of regions.
Since simple $K$-constraints are also $K$-constraints, we have that $\mathsf{Reg}(\Sigma,K) \subset \SC(\Sigma,K)$.

\begin{definition}[ERA]
\label{def:dera}
    A $K$-Event-Recording Automaton ($K$-ERA) $A$ is a tuple $(Q,q_{\mathsf{init}},\Sigma,E,F)$ where $Q$ is a finite set of states, $q_{\sf init} \in Q$ is the initial state, $\Sigma$ is a finite alphabet, $E \subseteq Q \times \Sigma \times \SC(\Sigma,K) \times Q$ is the set of transitions, and $F \subseteq Q$ is the subset of accepting states. Each transition in $A$ is a tuple $(q, \sigma, g, q')$, where $q, q' \in Q$, $\sigma \in \Sigma$ and $g \in \SC(\Sigma, K)$, $g$ is called the \emph{guard} of the transition.
$A$ is called a $K$-\emph{deterministic}-ERA ($K$-DERA) if for every state $q \in Q$ and every letter $\sigma$, if there exist two transitions $(q, \sigma, g_1, q_1)$ and $(q, \sigma, g_2, q_2)$ then $\sem{g_1} \cap \sem{g_2} = \emptyset$.
\end{definition}

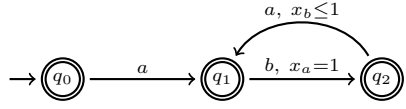
\begin{figure}[t]
    \centering
    \begin{tikzpicture}[scale=0.7]
    \everymath{\scriptstyle}
        \begin{scope}[every node/.style={circle, draw, inner sep=2pt,
            minimum size = 3mm, outer sep=3pt, thick}] 
            \node [double] (0) at (0,0) {$q_0$}; 
            \node [double] (1) at (3,0) {$q_1$};
            \node [double] (2) at (6,0) {$q_2$};
        \end{scope}
        \begin{scope}[->, thick]
            \draw (-1,0) to (0); 
            \draw (0) to (1);
            \draw (1) to (2);
            \draw [bend right = 60] (2) to (1);
        \end{scope}
        \node at (1.5,0.2) {$a$};
        \node at (4.5,0.25) {$b,~x_a = 1$};
        \node at (4.5,1.3) {$a,~x_b \leq 1$};
    \end{tikzpicture}
    \caption{A DERA that accepts timed words where every $a$ is followed by a $b$ after exactly $1$ time unit and every $b$ is followed by an $a$ within $1$ time units}
    \label{fig:dera}
\end{figure}

For the semantics, initially, all the clocks start with the value $0$ and then they all elapse at the same rate. For every transition on a letter $\sigma$, once the transition is taken, its corresponding recording clock $x_\sigma$ gets reset to the value $0$.

\subparagraph*{Clocked words.} \label{paragraph:clocked-words}
Given a timed word ${\sf tw}=(\sigma_1,t_1)(\sigma_2,t_2)\dots (\sigma_n,t_n)$, we associate with it a \emph{clocked word} 
\label{timed-to-clocked-word}
$\mathsf{cw}(\mathsf{tw})=(\sigma_1,v_1)(\sigma_2,v_2)\dots (\sigma_n,v_n)$ where each $v_i : X_{\Sigma} \rightarrow \mathbb{R}_{\geq 0}$ maps each clock of $X_{\Sigma}$ to a real-value as follows: $v_i(x_\sigma)=t_i-t_j$ where $j=\max\{ l < i \mid \sigma_l=\sigma \}$, with the convention that $\max(\emptyset)=0$. 
In words, the value $v_i(x_\sigma)$ is the amount of time elapsed since the last occurrence of $\sigma$ in {\sf tw}; 
which is why we call the clocks $x_\sigma$ `recording' clocks. Although not explicitly, clocks are implicitly reset after every occurrence of their associated events.

\subparagraph*{Timed language of a $K$-ERA.}
A timed word $\mathsf{tw}=(\sigma_1,t_1)(\sigma_2,t_2)\dots (\sigma_n,t_n)$ with its clocked word ${\sf cw}({\sf tw})=(\sigma_1,v_1)(\sigma_2,v_2)\dots (\sigma_n,v_n)$, is \emph{accepted} by $A$ if there exists a sequence of states $q_0 q_1 \dots q_n$ of $A$ such that $q_0=q_{\sf init}$, $q_n \in F$, and for all $1 \leq i \leq n$, there exists $e=(q_{i-1},\sigma,\psi,q_i) \in E$ such that $\sigma_i=\sigma$, and $v_i \models \psi$. 
The set of all timed words accepted by $A$ is called the \emph{timed language} of $A$, and will be denoted by $\ltw{A}$. 
We now ask a curious reader: is it possible to construct an ERA with two (or less) states that accepts the timed language represented in Figure~\ref{fig:dera}? We will answer this question in Section~\ref{section:implementation}.

\begin{lemma}[\cite{AFH99}]
\label{lem:era-equals-dera}
    For every $K$-ERA $A$, there exists a $K$-DERA $A'$ such that $\ltw{A} = \ltw{A'}$ and there exists a $K$-DERA $\overline{A'}$ such that $\overline{\ltw{A}} = \ltw{\overline{A'}}$.
\end{lemma}

A timed language $L$ is $K$-ERA (resp. $K$-DERA) definable if there exists a $K$-ERA (resp. $K$-DERA) $A$ such that $\ltw{A} = L$. Now, due to Lemma~\ref{lem:era-equals-dera}, a timed language $L$ is $K$-ERA definable iff it is $K$-DERA definable.

\subparagraph*{Symbolic words}
over $(\Sigma,K)$ are finite sequences $(\sigma_1,g_1)(\sigma_2,g_2)\dots (\sigma_n,g_n)$ where each $\sigma_i \in \Sigma$ and $g_i \in \SC(\Sigma,K)$, for all $1 \leq i \leq n$.
Similarly, a \emph{region word} over $(\Sigma,K)$ is a finite sequence $(\sigma_1,r_1)(\sigma_2,r_2)\dots (\sigma_n,r_n)$ where each $\sigma_i \in \Sigma$ and $r_i \in {\sf Reg}(\Sigma,K)$ is a simple $K$-constraint, for all $1 \leq i \leq n$. Lemma~\ref{lem:simple-region} will justify why we refer to symbolic words over simple constraints as `region' words.

We are now equipped to define when a timed word ${\sf tw}$ is compatible with a symbolic word $\sw$. Let  $\tw=(\sigma_1,t_1)(\sigma_2,t_2)\dots (\sigma_n,t_n)$ be a timed word with its clocked word being $(\sigma_1, v_1)(\sigma_2, v_2)\dots (\sigma_n, t_n)$, and let $\sw=(\sigma'_1,g_1)(\sigma'_2,g_2)\dots (\sigma'_m,g_m)$ be a symbolic word. We say that $\tw$ is compatible with $\sw$, noted $\tw \models \sw$, if we have that: $(i)$ $n=m$, \emph{i.e.}, both words have the same length, $(ii)$ $\sigma_i=\sigma'_i$ and $(iii)$ $v_i \models g_i$, for all $i$, $1 \leq i \leq n$. We denote by $\sem{\sw}=\{ \tw \in {\sf TW}(\Sigma) \mid \tw \models \sw\}$. We say that a symbolic word $\sw$ is \emph{consistent} if $\sem{\sw} \neq \emptyset$, and \emph{inconsistent} otherwise. 
We will use $\regL$ to denote the set of all \emph{consistent} region words over the set of all $K$-simple constraints $\mathsf{Reg}(\Sigma, K)$.

\begin{example}
Let $\Sigma = \{a,b\}$ be an alphabet and let $X_{\Sigma} = \{x_a, x_b\}$ be its set of recording clocks.
    The symbolic word $\rw_1 := (a,~x_a=0\wedge x_b = 0)(b,~x_a=1 \wedge x_b=1)$ is consistent, since the timed word $(a,0)(b,1) \models \rw_1$. On the other hand, consider the symbolic word $\rw_2 := (a,~x_a=0\wedge x_b = 1)(b,~x_a=1 \wedge x_b = 0)$. Let $\tw = (a,t_1)(b,t_2)$ be a timed word and $\cw(\tw) = (a,v_1)(b,v_2)$ be the clocked word of $\tw$. Then $\tw \models \rw_2$ would imply $t_1 = 0$ and $t_2 - t_1 = 1$, but in that case, $v_2(x_b) = 1$, that is, $v_2 \not\models x_a=1\wedge x_b = 0$. Hence, $\rw_2$ is inconsistent.
\end{example}

\begin{lemma}
\label{lem:rw-for-tw}
    For every timed word $\tw$ over an alphabet $\Sigma$, for every fixed positive integer $K$, there exists a unique region word $\rw \in \regL$ where $\tw \in \sem{\rw}$.
\end{lemma}

\begin{proof}[Proof sketch]
    Given a timed word $\tw$, consider its clocked word $\cw(\tw)$ (as in Page~\pageref{paragraph:clocked-words}). Now, it is easy to see that each clock valuation satisfies only one simple $K$-constraint. Therefore, we can construct a region word by replacing the valuations present in $\cw(\tw)$ with the simple $K$-constraints they satisfy, thereby constructing the (unique) region word $\rw$ that $\tw$ is compatible with.
\end{proof}

\begin{lemma}
\label{lem:simple-region}
For every $K$-ERA $A$ over an alphabet $\Sigma$ and for every region word $\rw \in \regL$, either $\sem{\rw} \subseteq \ltw{A}$ or $\sem{\rw} \cap \ltw{A}=\emptyset$. Equivalently, for all timed words $\tw_1,\tw_2 \in \sem{\rw}$, $\tw_1 \in \ltw{A}$ iff $\tw_2 \in \ltw{A}$.
\end{lemma}

The above lemma states that, two timed words that are both compatible with one region word, cannot be `distinguished' using a DERA. This  can be proved using Lemma~$18$ of~\cite{GJL10}.

\subparagraph*{Symbolic languages.}
\label{page:zone-era-to-regions}
Note that, every $K$-DERA $A$ can be transformed into another $K$-DERA $A'$ where every guard present in $A'$ is a simple $K$-constraint. This can be constructed as follows: for every transition $(q, \sigma, g, q')$ in $A$ where $g$ is a $K$-constraint (and not a simple $K$-constraint), $A'$ contains the transitions $(q, \sigma, r, q')$ such that $r \models g$. 
Since, for every $K$-constraint $g$, there are only finitely-many (possibly, exponentially many) simple $K$-constraints $r$ that satisfy $g$, this construction indeed yields a finite automaton. As for an example, the transition from $q_1$ to $q_2$ in the automaton in Figure~\ref{fig:dera} contains a non-simple $K$-constraint as guard. This transition can be replaced with the following set of transitions without altering its timed language: $(q_1, b, x_a=1 \wedge x_b = 0, q_2), (q_1, b, x_a=1 \wedge 0 < x_b < 1, q_2), (q_1, b, x_a=1 \wedge x_b = 1, q_2), (q_1, b, x_a=1 \wedge x_b > 1, q_2)$. 
Note that, $A'$ in the above construction is different from the ``simple DERA'' for $A$, as defined in~\cite{GJL10}. A key difference is that,  
in $A'$ there can be paths leading to an accepting state that are not satisfiable by any timed words, which is not the case in the simple DERA of $A$. This relaxation ensures, $A'$ here has the same number of states as $A$, in contrast, in a simple DERA of~\cite{GJL10}, the number of states can be exponentially larger than the original DERA.

With the above construction in mind, given a $K$-DERA $A$, we associate two regular languages to it, called the {\em syntactic region language} of $A$ and the {\em region language} of $A$, both over the finite alphabet $\Sigma \times {\sf Reg}(\Sigma,K)$, as follows:
\label{page:lrw-lsrw}
  \begin{itemize}
        \item the syntactic region language of $A$, denoted $L^{\sf s \cdot rw}(A)$, is the set of region words ${\sf rw}=(\sigma_1,r_1)(\sigma_2,r_2)\dots (\sigma_m,r_m)$ such that there exists a sequence of states $q_0 q_1 \dots q_n$ in $A$ and $(i)$ $q_0=q_{\sf init}$, $q_n \in F$, and for all $1 \leq i \leq n$, there exists $e=(q_{i-1},\sigma,\psi,q_{i}) \in E$ such that $\sigma_i=\sigma$, and $r_i \models \psi$. 
        It is easy to see that $A$ plays the role of a DFA for this language and that this language is thus a regular language over $\Sigma \times {\sf Reg}(\Sigma,K)$.
        \item the (semantic) region language of $A$, denoted $L^{\sf rw}(A)$, is the subset of $L^{\sf s \cdot rw}(A)$ that is restricted to the set of region words that are consistent. This language is thus the intersection of  $L^{\sf s \cdot rw}(A)$ with ${\sf RegL}(\Sigma,K)$. And so, it is also a regular language over $\Sigma \times {\sf Reg}(\Sigma,K)$.
        (This is precisely the language accepted by a ``simple DERA'' corresponding to $A$, as defined in~\cite{GJL10}).
  \end{itemize}

\begin{example}
Let $A$ be the automaton depicted in Figure~\ref{fig:dera}.  
    Now, consider the region word $w_r := (a, x_a = 1 \wedge x_b = 1)(b, x_a = 1 \wedge x_b = 3)$. Note that, $\rw \in \lsrw{A}$, since there is the path $q_0 \xrightarrow[]{a,~\top} q_1 \xrightarrow[]{b,~x_a = 1} q_2$ in $A$ and (i) $(x_a = 1 \wedge x_b = 1) \models \top$, and (ii) $(x_a = 1 \wedge x_b = 3) \models (x_a = 1)$. 
    However, one can show that $\rw \notin \lrw{A}$. Indeed, $\sem{w_r} = \emptyset$, since for every clocked word $\mathsf{cw} = (a, v_1)(b, v_2)$, $\mathsf{cw} \models \rw$ only if $v_2(x_a) = 1 $ and $v_2(x_b) = 3$,
however, this is not possible, since in order for $v_2(x_a)$ to be $1$, the time difference between the $a$ and $b$ must be $1$-time unit, in which case, $v_2(x_b)$ will always be $2$ (because, $v_1(x_b) = 1$).
\end{example}

\begin{remark}
Given a DFA $C$ over the alphabet $\sevents$, one can interpret it as a K-DERA and associate the two --syntactic and semantic-- languages as defined above. Then, the regular language $L(C)$ is  the set of all region words in $L^{\sf s \cdot rw}(C)$, when $C$ is interpreted as a DERA. Similarly, abusing the notation, we write $L^{\sf rw}(C)$ to denote the set of all consistent region words accepted by $C$, and $L^{\sf tw}(C)$ to denote the set of all timed words accepted by $C$. This is a crucial remark to note, as \emph{we will use this notation often} in the rest of the paper.
\end{remark}

\begin{table}[t]
    \centering
    \begin{tabular}{c||c|c|c}
        symbolic word & $L(A)$ & $\lsrw{A}$ & $\lrw{A}$ \\
        \hline
        \hline
        $(a,\top)$ & \tick & \cross & \cross \\
        \hline
        $(a,~x_a=1\wedge x_b = 1)(b,~x_a=0\wedge x_b = 1)$ & \cross & \tick & \tick \\
        \hline
        $(a,~x_a>1\wedge x_b > 1)(b,~x_a>1 \wedge x_b = 0)$ & \cross & \tick & \cross \\
        \hline
    \end{tabular}
    \caption{Different symbolic languages for the automaton $A$ in Figure~\ref{fig:dera}. 
    }
    \label{tab:languages}
\end{table}

\section{greybox learning framework}
\label{sec:learning-algo}
In this work, we are given a target timed language $L$, and we are trying to infer a DERA that recognizes $L$. To achieve this, we propose a greybox learning algorithm that involves a {\tt Teacher} and a {\tt Student}.
Before giving details about this algorithm, we first make clear what are the assumptions that we make on the interaction between the {\tt Teacher} and the {\tt Student}. We then introduce some important notions that will be useful when we present the algorithm in Section~\ref{sec:tlsep}.

\subsection{The learning protocol}
\label{paragraph:learning-protocol}
The purpose of the greybox active learning algorithm that we propose is to build a DERA for an ERA-definable timed language $L$ using the following protocol between a so-called {\tt Teacher}, that knows the language, and {\tt Student} who tries to discover it. 
We assume that the {\tt Student} is aware of the alphabet $\Sigma$ and the maximum constant $K$ that can appear in a region word, thus,
due to Lemma~\ref{lem:regL-is-reg}, the {\tt Student} can then determine the set of \emph{consistent} region words $\regL$. 

\begin{lemma}
\label{lem:regL-is-reg}
Given an alphabet $\Sigma$ and a positive integer $K$, the set ${\sf RegL}(\Sigma,K)$ consisting of all consistent region words over $(\Sigma, K)$, forms a regular language over the alphabet $\Sigma \times {\sf Reg}(\Sigma,K)$ and it can be recognized by a finite automaton.
In the worst case, the size of this automaton is exponential both in $|\Sigma|$ and $K$.    
\end{lemma}

\begin{proof}[Proof sketch]
    Consider the one-state ERA where the state contains a transition to itself on every pair from $\Sigma \times \mathsf{Reg}(\Sigma, K)$. Now, due to Lemma~\ref{lem:era-equals-dera}, we know there exists a DERA $A^d$ such that $\ltw{A} = \ltw{A^d}$. Moreover, due to results presented in~\cite{GJL10}, we know $A^d$ can be transformed into a ``simple DERA'' $A^{sd}$ such that $\ltw{A^d} = \ltw{A^{sd}}$. Now, from the definition of simple DERA, we get that for every region word $\rw \in L(A^{sd})$, $\sem{\rw} \neq \emptyset$, i.e. $L(A^{sd}) \subseteq \regL$. On the other hand, since the language of $A$ is the set of all timed words, this means, $L(A^{sd})$ must contain the set of all consistent region words which is $\regL$. Therefore, we get that $L(A^{sd}) = \regL$. The size of $A^{sd}$ is exponential in both $|\Sigma|$ and $K$  (\cite{GJL10}). 
\end{proof}

Since the {\tt Student} can compute the automaton for $\regL$, we assume that the {\tt Student} only poses membership queries for \emph{consistent} region words, that is, region words from the set $\regL$.
The {\tt Student}, that is the learning algorithm, is allowed to formulate two types of queries to the {\tt Teacher}:

\smallskip
{\bf Membership queries}:
    given a \emph{consistent} region word $\rw$ over the alphabet $\Sigma \times \mathsf{Reg}(\Sigma, K)$, {\tt Teacher} answers {\tt Yes} if $\sem{\rw} \subseteq L$, 
    and {\tt No} if $\sem{\rw} \not\subseteq L$ which, for region words, is equivalent to $\sem{\rw} \cap L=\emptyset$ (cf. Lemma~\ref{lem:simple-region}). 

 {\bf Equivalence queries}: 
    given a DERA $A$, {\tt Student} can ask two types of inclusion queries: $\ltw{A} \subseteq L$, or $L \subseteq \ltw{A}$. The {\tt Teacher} answers {\tt Yes}  to an inclusion query if the inclusion holds, otherwise returns a counterexample.

\subparagraph*{Comparison with the work in~\cite{LADSL11}.}
The queries that are allowed in our setting are {\em semantic} queries, so, they can be answered if the {\tt Teacher} knows the timed language $L$ and the answers to the queries are not bound by the particular ERA that the {\tt Teacher} has for the language.
This is in contrast with the learning algorithm proposed in~\cite{LADSL11}. Indeed, the learning algorithm proposed there considers queries that are answered for a specific DERA that the {\tt Teacher} uses as a reference to answer those queries. In particular, a membership query in their setting is formulated according to a symbolic word $w_z$, i.e. a finite word over the alphabet $\Sigma \times \SC(\Sigma,K)$. 
The membership query about the symbolic word $w_z$ will be answered positively if there is a run that follows edges exactly annotated by this symbolic word in the automaton and that leads to an accepting state. For the equivalence query, the query asks if a candidate DERA is equal to the DERA that is known by the {\tt Teacher}. Again, the answer is based on the syntax and not on the semantics. For instance, 
consider the automaton $A$ in Figure~\ref{fig:dera}. When a membership query is performed in their algorithm for the symbolic word $w_z := (a,\top)(b,x_a = 1)(a,x_a = 1)$, even though $\sem{w_z} \subseteq \ltw{A}$, their {\tt Teacher} will respond negatively -- this is because the guard on the transition $q_2 \xrightarrow{a, x_a \leq 1} q_1$ in $A$, does not syntactically match the expression $x_a = 1$, which is present in the third position of $w_z$. That is, when treating $A$ as a DFA with the alphabet being $\Sigma \times \SC(\Sigma, K)$, $w_z \notin L(A)$. 
So, the algorithm in~\cite{LADSL11} essentially learns the regular languages of symbolic words over the alphabet $\Sigma \times \SC(\Sigma,K)$ rather than the semantics of the underlying timed language. As a corollary, our result (which we develop in Section~\ref{sec:tlsep}) on returning a DERA with the minimum number of states holds even if the input automaton has more states, contrary to~\cite{LADSL11}.

\subsection{The region language associated to a timed language} 

\begin{definition}[The region language]
\label{def:region-lang}
    Given a timed language $L$, we define its $K$-region language as the set $\RW{L} := \{\rw \in \regL \mid \sem{\rw} \subseteq L\}$.
\end{definition}
For every ERA-definable timed language $L$, its $K$-region language $\RW{L}$ is uniquely determined due to Lemma~\ref{lem:simple-region}.
We now present a set of useful results regarding the set $\RW{L}$.
\begin{lemma}
\label{lem:timed-language-union-of-regions}
    Let $L$ be an ERA-definable timed language, and let $A$ be a $K$-DERA recognizing $L$. Then, the following statements hold:
    \begin{enumerate}
        \item $\bigcup_{\rw \in \RW{L}} \sem{\rw} = L$, 
        \item $\RW{\overline{L}} = \overline{\RW{L}} \cap \regL$, where $\overline{L}$ denotes the complement of $L$\footnote{The complement of $L$ is also ERA-definable, since the class of ERA is closed under complementation~\cite{AFH99}},
    \end{enumerate}
\end{lemma}

\begin{proof}
    \begin{enumerate}
        \item From Definition~\ref{def:region-lang}, we get that $\bigcup_{\rw \in \RW{L}} \sem{\rw} \subseteq L$. Conversely, let $\tw \in L$. From Lemma~\ref{lem:rw-for-tw} we know there exists a unique $\rw \in \regL$ such that $\tw \in \sem{\rw}$. Then we know $\sem{\rw} \cap L \neq \emptyset$, which implies (using Lemma~\ref{lem:simple-region}) $\sem{\rw} \subseteq L$. Therefore, $\rw \in \RW{L}$. This shows first equality in the statement.
        The second equality holds since $A$ recognizes $L$.
        For the third equality, again let $\tw \in \ltw{A}$, then consider the unique region word $\rw$ such that $\tw \in \sem{\rw}$. It must be the case that $\rw \in \lrw{A}$. Hence, we get $\ltw{A} \subseteq \bigcup_{\rw \in \lrw{A}} \sem{\rw}$. On the other hand, let $\tw \in \sem{\rw}$ for a $\rw \in \lrw{A}$.
        Then, clearly $\tw \in \ltw{A}$.
        \item Let $\rw \in \RW{\overline{L}}$, then from Definition~\ref{def:region-lang}, $(\emptyset \neq) \sem{\rw} \subseteq \overline{L}$. Since $L \cap \overline{L} = \emptyset$, we get $\sem{\rw} \cap L = \emptyset$ and hence $\rw \in \overline{\RW{L}} \cap \regL$.
    \end{enumerate}\end{proof}

In~\cite{GJL10}, the authors prove that for every ERA-definable timed language $L$, there exists a $K$-DERA $A$, for some positive integer $K$, such that $\ltw{A} = L$. From the construction of such an $A$, it in fact follows that $\lrw{A} = \RW{L}$. The authors then present an extension of the $L^*$-algorithm~\cite{Angluin87} that is able to learn a minimal DFA for the language $\lrw{A}$. 
Now, the following lemma shows that a minimal DFA for the regular language $\lrw{A}$ can be \emph{exponentially} larger than a minimal DFA for the language $\lsrw{A}$.
In this paper, we show how to cast the problem of learning a minimal DFA for $L^{\sf s \cdot rw}(A)$ using an adaptation of the greybox active learning framework proposed by Chen et al. in~\cite{CFCTW09}. 

\begin{lemma}
\label{lem:srw-to-rw}
 \label{lem:complement-rw}
 Let $L$ be an ERA-definable timed language and $A$ be 
a $K$-DERA recognizing $L$, then the following statements hold:
 \begin{enumerate}
     \item 
 $\RW{L}=L^{\sf s \cdot rw}(A) \cap {\sf RegL}(\Sigma,K) = \lrw{A}$,
 \item
 $\RW{\overline{L}} = \lsrw{\overline{A}} \cap \regL = \lrw{\overline{A}}$,
 \item The minimal DFA over $\Sigma \times {\sf Reg}(\Sigma,K)$ that recognizes $L^{\sf s \cdot rw}(A)$ can be exponentially smaller than the minimal DFA that recognizes $L^{\sf rw}(A)$,  
 \item There exists a DFA with a number of states equal to the number of states in $A$ that accepts $L^{\sf s \cdot rw}(A)$.
 \end{enumerate}
\end{lemma}

\begin{proof}
    \begin{enumerate}
        \item The second equality comes from the definition of the sets $\lrw{A}$ and $\lsrw{A}$ (cf. Page~\pageref{page:lrw-lsrw}). For the first equality, since $A$ recognizes $L$, we know that for every $\rw \in \RW{L}$, $\rw \in \lrw{A}$ and hence (from the definition of $\lrw{A}$) $\rw \in \lsrw{A} \cap \regL$. On the other hand, 
        if $\rw \in \lsrw{A} \cap \regL$, then $\rw \in \lrw{A}$ and thus $\sem{\rw} \subseteq L$.
        \item This can be proved similarly to the previous result.
        \item We know that checking emptiness of the timed language of an ERA is $\mathsf{PSPACE}$-complete~\cite{AFH99}.
        Now, we can reduce the emptiness checking problem of an ERA $A$ to checking if $\lrw{A}$ is empty or not. Note that, from Lemma~\ref{lem:rw-for-tw} and~\ref{lem:simple-region}, $\ltw{A} = \emptyset$ iff $\lrw{A} = \emptyset$. Since it is possible to decide in polynomial time if the language of a DFA is empty or not, this implies, the DFA recognizing $\lrw{A}$ must be exponentially larger (in the worst case) than the DFA recognizing $\lsrw{A}$.
        \item Note that, $\lsrw{A}$ contains inconsistent region words. Now, $L(A)$ may contain symbolic words (and not only region words). Following the construction presented in Page~\pageref{page:zone-era-to-regions}, we know that from $A$ we can construct another $K$-DERA $A'$ with the same number of states as that of $A$ such that (i) $A'$ contains only $K$-simple constraints as its guards and (ii) $\ltw{A} = \ltw{A'}$. It can then be verified easily that $L(A') = \lsrw{A}$. This is because, for every region word $\rw \in \lsrw{A}$, there exists a path in $A'$ corresponding to $\rw$ (that is, transitions marked with the letter-simple constraint pairs in $\rw$) leading to an accepting state in $A'$.
    \end{enumerate}
\end{proof}

Note that, the first equation in Lemma~\ref{lem:srw-to-rw} is similar to Equation~\ref{eq:intersection} that we introduced in Section~\ref{sec:intro}. This equation is the cornerstone of our approach, and this is what justifies the advantage of our approach over existing approaches.
It is crucial to understand that for a fixed $\Sigma$ and $K$, the regular language ${\sf RegL}(\Sigma,K)$ present in this equation is already known and does not require learning. Consequently, the learning method we are about to describe focuses on discovering an automaton $C$ that satisfies this equation, and thereby it is a {\em greybox} learning algorithm (as coined in~\cite{DBLP:conf/nfm/0002R16}) since a part of this equation is already known.

In order for an automaton $C$ to satisfy the first equation in Lemma~\ref{lem:srw-to-rw}, we can deduce that $C$ must (i) accept all the region words in $\RW{L}$, $\ie$ all consistent region words $\rw$ for which $\sem{\rw} \subseteq L$, (ii) reject all the consistent region words $\rw$ for which $\sem{\rw} \subseteq \overline{L}$, and (iii) it can do either ways with inconsistent region words. 
This is precisely the flexibility that allows for solutions $C$ that have fewer states than the minimal DFA for 
the language $\RW{L}$.
The fourth statement in Lemma~\ref{lem:srw-to-rw} implies that there is a 
DFA over the alphabet $\Sigma \times {\sf Reg}(\Sigma,K)$ that has the same number of states as that of a minimal DERA recognizing the timed language $L$. 
We give an overview of our greybox learning algorithm in Section~\ref{sec:tlsep}.

In the following, we establish a strong correspondence between the timed languages and the region languages of two DERAs:
\begin{lemma}
    \label{lem:rw-iff-tw}
    Given two DERAs $A$ and $B$, $\lrw{A} \subseteq \lrw{B}$ iff $\ltw{A} \subseteq \ltw{B}$.
\end{lemma}

\begin{proof}
    $(\Rightarrow)$ This direction follows from Lemma~\ref{lem:timed-language-union-of-regions}.
    \smallskip

    $(\Leftarrow)$
    Let $w \in \lrw{A}$, then we know that $\sem{w} \neq \emptyset$. Choose a timed word $\tw \in \sem{w}$. Clearly, $\tw \in \ltw{A}$ and from the assumption of the lemma, we get $\tw \in \ltw{B}$. From Lemma~\ref{lem:rw-for-tw} we know that $w$ is the only region word such that $\tw \in \sem{w}$. Finally, since $\tw \in \ltw{B}$, we can deduce from Lemma~\ref{lem:simple-region} $w \in \lrw{B}$.
\end{proof}

\subsection{(Strongly-)complete 3DFAs} 

\begin{figure}[t]
  \centering
  \begin{subfigure}{.45\textwidth}
  \centering
      \includegraphics[scale=0.17]{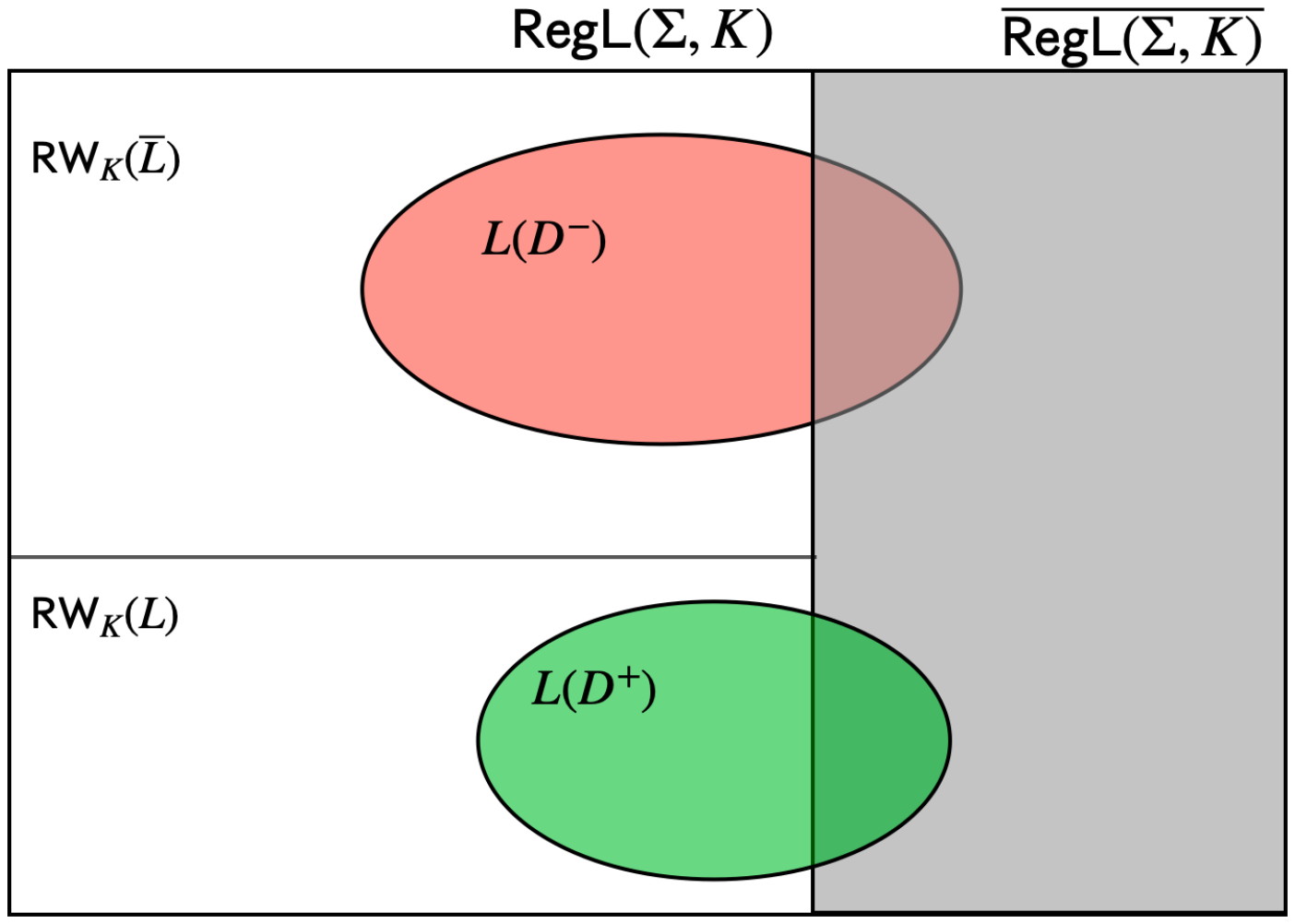}
      \caption{Completeness}
      \label{fig:illustration-completeness}
  \end{subfigure}
  \begin{subfigure}{.45\textwidth}
  \centering
      \includegraphics[scale=0.17]{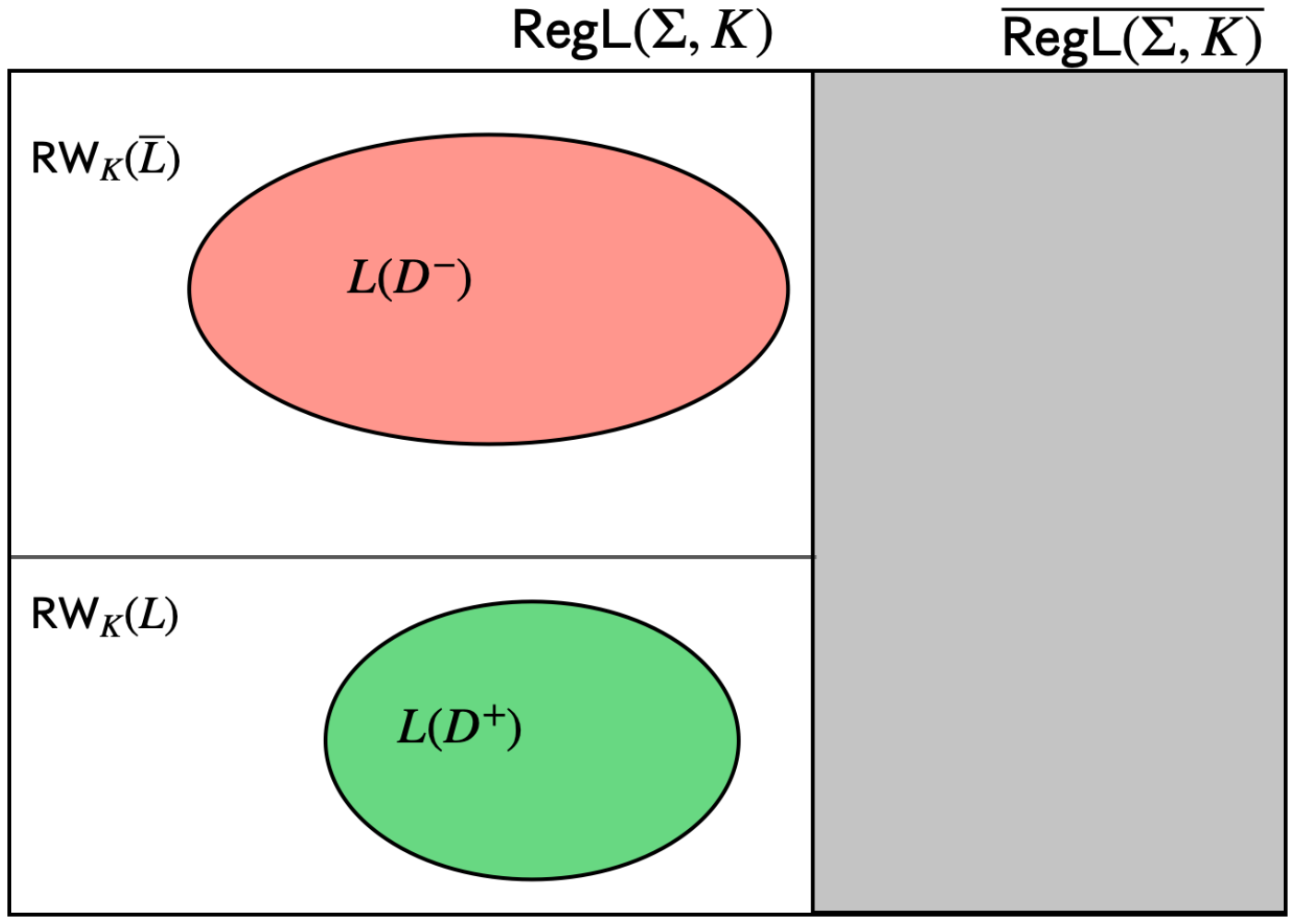}
      \caption{Strong-completeness}
      \label{fig:s-complete}
  \end{subfigure}
\caption{Illustration of the two completeness criteria}
\label{fig:comp-and-s-comp}
\end{figure}
In this part, we define some concepts regarding 3DFAs that we will use in our greybox learning algorithm. These concepts are adaptations of similar concepts presented in~\cite{CFCTW09} to the timed setting.
  \begin{definition}[Completeness]
      \label{def:completeness}
       A 3DFA $D$ over $\sevents$ is {\em complete} w.r.t. the timed language  $L$ if
           the timed language of $D^+$, when interpreted as a DERA, is a subset of $L$,
       \ie, $\ltw{D^+} \subseteq L$; and
       the timed language of $D^-$, when interpreted as a DERA, is a subset of $\overline{L}$, \ie, 
        $\ltw{D^-} \subseteq \overline{L}$.
  \end{definition}
Applying Lemma~\ref{lem:srw-to-rw} and~\ref{lem:rw-iff-tw} to the above definition, we get the following result:
\begin{lemma}
\label{lem:complete-for-region}
Let $L$ be a timed language over $\Sigma$ recognized by a $K$-DERA $A$ over $\Sigma$, then for any 3DFA $D$ over $\sevents$,
    $D$ is complete w.r.t. $L$ iff $\lrw{D^+} \subseteq \lrw{A} = \RW{L}$ and $\lrw{D^-} \subseteq \lrw{\overline{A}} = \RW{\overline{L}}$.
\end{lemma}
\begin{proof}
    Since $A$ accepts $L$, so $\ltw{A} = L$, and since $A$ is a DERA, $\ltw{\overline{A}} = \overline{L}$.
    Therefore, $\ltw{D^+} \subseteq L$ iff $\ltw{D^+} \subseteq \ltw{A}$ iff $\lrw{D^+} \subseteq \lrw{A}$ (using Lemma~\ref{lem:rw-iff-tw}).
    Similarly, $\ltw{D^-} \subseteq \overline{L}$ iff $\ltw{D^-} \subseteq \ltw{\overline{A}}$ iff $\lrw{D^-} \subseteq \lrw{\overline{A}}$. 
    Then from Lemma~\ref{lem:srw-to-rw}, we know that $\lrw{A} = \RW{L}$, and also $\lrw{\overline{A}} = \RW{\overline{L}}$, therefore, the statement follows.
\end{proof}

However, to show the minimality in the number of states of the final DERA returned by our greybox learning algorithm, we introduce \emph{strong-completeness}, by strengthening the condition on the language accepted by a 3DFA $D$ as follows: 
$L(D^+) = \lrw{D^+}$, and $L(D^-) = \lrw{D^-}$.
In other words, we require that all the inconsistent words lead to a ``don't care'' state in $D$. 
       
  \begin{definition}[Strong-completeness]
      \label{def:s-completeness}
      A 3DFA $D$ over $\sevents$ is {\em strongly-complete} w.r.t. the timed language  $L$ if $D$ is complete w.r.t. $L$ and additionally, every region word that is strongly accepted or strongly rejected by $D$ must be consistent, \ie, if
           $\ltw{D^+} \subseteq L$ and $L(D^+) = \lrw{D^+}$; and
       $\ltw{D^-} \subseteq \overline{L}$ and $L(D^-) = \lrw{D^-}$.
  \end{definition}

\begin{corollary}
\label{cor:s-complete-for-region}
    Let $L$ be a timed language over $\Sigma$ recognized by a $K$-DERA $A$ over $\Sigma$, then for any 3DFA $D$ over $\sevents$,
    $D$ is strongly-complete w.r.t. $L$ iff $L(D^+) \subseteq \lrw{A} = \RW{L}$ and $L(D^-) \subseteq \lrw{\overline{A}} = \RW{\overline{L}}$.
    $D$ is strongly-complete w.r.t. $L$ iff $L(D^+) \subseteq \RW{L}$ and $L(D^-) \subseteq \RW{\overline{L}}$.
\end{corollary}
 
 A pictorial representation of completeness and strong-completeness are given in Figure~\ref{fig:comp-and-s-comp}.
 A dual notion of completeness, called \emph{soundness} can also be defined in this context, however, since we will only need that notion for showing the termination of \tlsep, we will introduce it in the next section.

\section{The \tlsep\ algorithm}
\label{sec:tlsep}
To learn a ERA-recognizable language $L$, we adopt the following greybox learning algorithm, that is an adaptation of the separation algorithm proposed in~\cite{CFCTW09}. The main steps of the algorithm are the following:
    we first learn a (strongly-)complete 3DFA $D_i$ that maps finite words on the alphabet $\Sigma \times {\sf Reg}(\Sigma,K)$ to $\{0,1,?\}$, where $1$ means ``accept'', $0$ means ``reject'', and $?$ means ``don't care''. The 3DFA $D_i$ is used to represent a family of regular languages $L$ such that $L \models D_i$.
    We then extract from $D_i$ a minimal automaton $C_i$ that is compatible with the 3DFA $D_i$, \emph{i.e.}, such that $L(C_i) \models D_i$.
    Finally, we check, using two inclusion queries, if $C_i$ satisfies the equation $\ltw{A} = \ltw{C_i}$. If this is the case then we return $C_i$, otherwise, we use the counterexample to this equality check to learn a new 3DFA $D_{i+1}$.
We now give details of the greybox in the following subsections, and we conclude this section with the properties of our algorithm. 

\subsection{Learning the 3DFA $D_i$}
\label{sec:table-to-3DFA}
In this step, we learn the 3DFA $D_i$, by relying on a modified version of $L^*$. Similar to $L^*$, we maintain an observation table $(S, E, T)$, where $S$ (the set of `prefixes') and $E$ (the set of `suffixes') are sets of region words. $T : (S.E) \to \{0,1,{?}\}$ is a function that maps the region words $s.e$, where $s \in S, e \in E$, to $1$ if $\sem{s.e} \subseteq L$, $0$ if $\sem{s.e} \cap L = \emptyset$ and to ${?}$ if $\sem{s.e} = \emptyset$.
Crucially, we restrict \emph{membership queries} exclusively to consistent region words. The rationale is that queries for inconsistent words are redundant, given our understanding that such words invariably yield a ``$?$'' response from the {\tt Teacher}. This approach streamlines the learning phase by reducing the number of necessary membership queries.
We can see the 3DFA $D_i$ defining a space of solutions for the automaton $C$ that must be discovered.

\begin{figure}[t!]
    \centering
    \includegraphics[width=\linewidth]{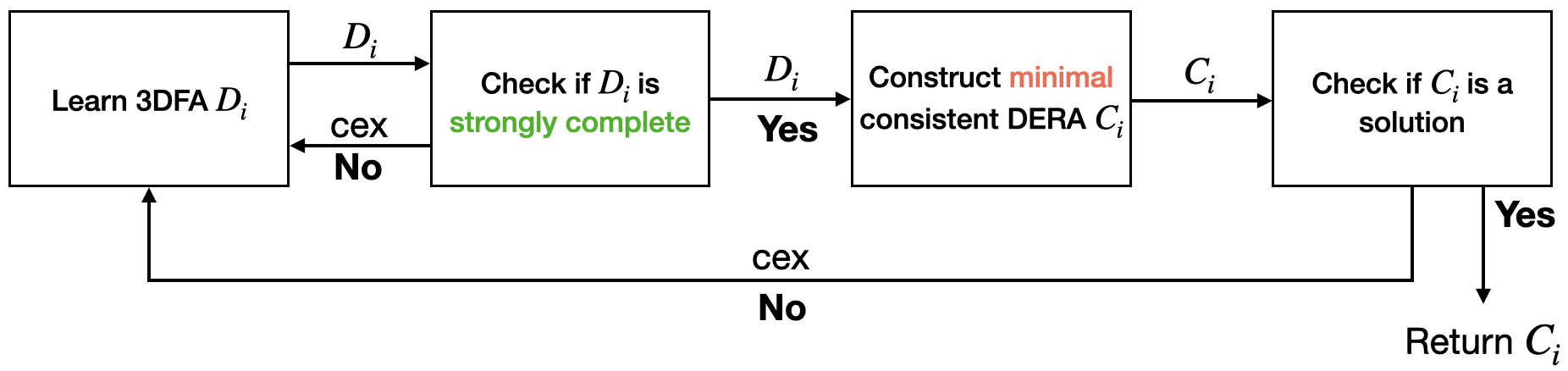}
    \caption{Overview of the \textsf{tLSep} algorithm}
    \label{fig:tlsep}
\end{figure}
\subsection{Checking if the 3DFA $D_i$ is (strongly-)complete}
\label{sec:comp-check}
We can interpret the 3DFA $D_i$ as the definition for a space of solutions for recognizing the region language $\RW{L}$.
In order to ensure that the final DERA of our algorithm has the minimum number of states, we rely on the fact that  at every iteration $i$, the 3DFA $D_i$ is \emph{strongly-complete}. 
  To this end, we show the following result:
  
\begin{lemma}
\label{lem:minimalility-3DFA}
    Let $L$ be a timed language, and $A = (Q,q_{\sf init},\Sigma,E,F)$ be a K-DERA with the \emph{minimum} number of states 
    accepting L.
    Let $C_i = (Q',q'_{\sf init},\sevents,\delta,F')$ be the minimal consistent DFA such that $L(C_i) \models D_i$. If $D_i$ is strongly-complete, then $|Q| \ge |Q'|$.
\end{lemma}
\begin{proof}
We shall show that $\lsrw{A}$ is consistent with $D_i$.
First, if $\word \in L(D_i^+)$, then $\word \in \lrw{D_i^+}$ (since strongly-complete), therefore from Corollary~\ref{cor:s-complete-for-region}, $\word \in \lrw{A}$ and hence $\word \in \lsrw{A}$.
For the other side, if $\word \in L(D_i^-)$, then $\word \in \lrw{D_i^-}$, and again from Corollary~\ref{cor:s-complete-for-region}, $\word \in \lrw{\overline{A}}$, which is equivalent to $\word \in \lsrw{\overline{A}} \cap \regL$. Now note that $\lsrw{\overline{A}} = \overline{\lsrw{A}}$, and hence $\word \in \overline{\lsrw{A}} \cap \regL$, and so, $\word \in \overline{\lsrw{A}}$.
Therefore, $\lsrw{A}$ is consistent with $D_i$.
From hypothesis, $C_i$ is the \emph{minimal} DFA that is consistent with $D_i$. Therefore, it follows that $|Q| \ge |Q'|$.
\end{proof}

The lemma above implies that if we compute a minimal consistent DFA $C_i$ from a strongly-complete $3$DFA $D_i$ w.r.t. the timed language $L$ and if this $C_i$ satisfies Equation~\ref{eq:intersection}, then $C_i$ recognizes (when it is interpreted as a DERA) $L$ and it has the minimum number of states among every DERA recognizing $L$. The two requirements -- $D_i$ being \emph{strongly}-complete and $C_i$ being a \emph{minimal} consistent DFA -- are necessary to ensure that the resulting DERA for $L$ has the minimum number of states. However, checking both of these requirements are computationally expensive. First, for \emph{strong}-completeness, one needs to check that no inconsistent region word is strongly accepted or strongly rejected by the 3DFA $D_i$. This can be done by checking if every region word that is either strongly accepted or strongly rejected by $D_i$ is accepted by a DFA recognizing the regular language $\overline{\regL}$ -- this is computationally hard due to the size of the latter.
Second, computing a minimal consistent DFA $C_i$ from $D_i$ requires essentially to find a DFA that accepts every region word that is strongly accepted by $D_i$ and rejects every region word that is strongly rejected by $D_i$. This is a generalization of the problem of finding a DFA with minimum number of states accepting a set of positive words and rejecting a set of negative words, which is known -- due to the works of Gold~\cite{DBLP:journals/iandc/Gold78} -- to be a computationally hard problem. 
In practice, we generally do not require (as is also remarked in~\cite{CFCTW09}) the automaton with the \emph{minimum} number of states, instead it is sufficient to compute an automaton with a \emph{small} number of states -- provided it makes the algorithm computationally easier.
To this end, instead of implementing the algorithm depicted in Figure~\ref{fig:tlsep}, we instead implement the algorithm depicted in Figure~\ref{fig:tlsep-heu} that employs two relaxations -- (i) we only compute a \emph{complete} $3$DFA $D_i$, and (ii) we only compute a \emph{small} (not necessarily minimal) $C_i$ consistent with $D_i$. Whereas checking whether a $3$DFA is complete or not can be done by checking the two inclusions mentioned in Definition~\ref{def:completeness}, we present a heuristic for computing a small consistent DFA in the next section. These two relaxations come at the cost of minimality of the solution, however, we show in Theorem~\ref{th:tlsep-correctness} that this  version of our algorithm still correctly computes a DERA recognizing $L$. As we will see in Section~\ref{section:implementation}, our implementation returns an automaton with a reasonably small, sometimes minimum, number of states.

\begin{figure}[t!]
    \centering
    \includegraphics[width=\linewidth]{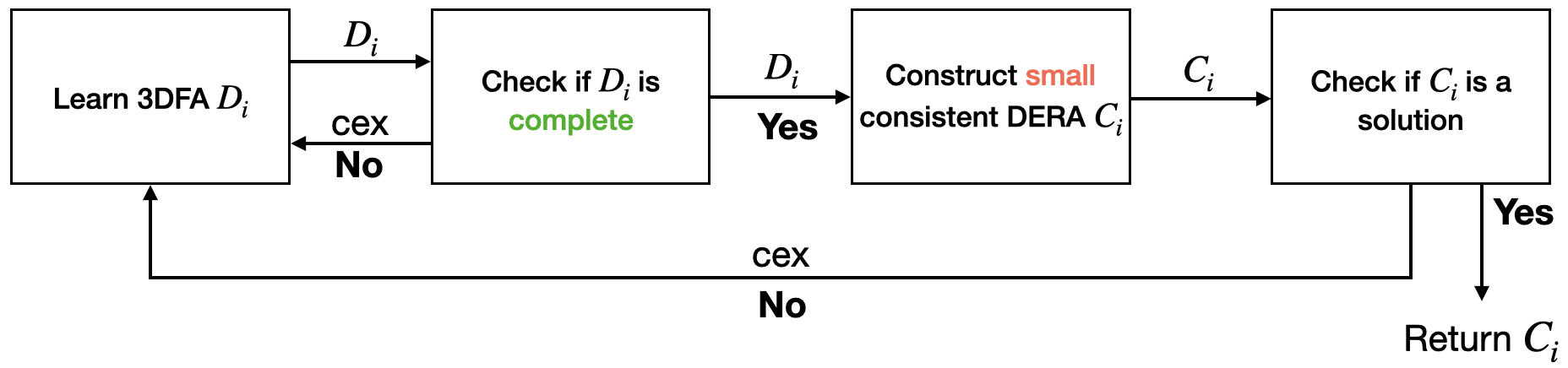}
    \caption{Overview of the \textsf{tLSep} algorithm with heuristics}
    \label{fig:tlsep-heu}
\end{figure}

\subsection{Computing a small candidate $C_i$ from $D_i$}

In this section, we detail a heuristic that computes a small candidate DFA $C_i$ from a $3$DFA $D_i$ that is complete with respect to $L$.
For the rest of this section, we fix the notation for the $3$DFA $D_i = (Q_i,q_{\sf init},(\sevents),\delta_i,{\sf A_i},{\sf R_i},{\sf E_i})$.
\medskip

    \noindent{\bf The first step} is to find the set of all pairs of states that are \emph{incompatible} in $D_i$.
  Two states are \emph{incompatible} if there exists a region word $\word \in (\sevents)^*$ s.t. $\word$ leads one of them to an accepting state in ${\sf A_i}$ and another to a rejecting state in ${\sf R_i}$. 
  Notice that, in the above definition, one does not take into account the timing constraints of a word, and that is precisely the reason why in the final automaton, the \emph{inconsistent} words may lead to any state, unlike in a ``simple DERA'' for the corresponding timed language, where the inconsistent words must be rejected syntactically by the automaton,  
  enabling it to be (possibly) smaller than the ``minimal simple DERA''. 

  We recursively compute the set $S_{\sf bad}$ of incompatible pairs of states as follows: initially add every pair $(q,q')$ into $S_{\sf bad}$ where $q\in {\sf A_i}$, $q' \in {\sf R_i}$; then recursively add new pairs as long as possible: add a new pair $(q_1,q_2)$ into $S_{\sf bad}$ s.t. there exists a letter $e \in \sevents$ s.t. $q_1 \xrightarrow{e} q_1'$ and $q_2 \xrightarrow{e} q_2'$ and $(q_1', q_2') \in S_{\sf bad}$.
\medskip

\noindent{\bf The second step} is to find the set of all $\subseteq$-\emph{maximal sets} that are compatible in $D_i$.
 A set of states of $D_i$ is \emph{compatible} if it does not contain any incompatible pair. 
  The set of maximal sets of compatible pairs, denoted by $S_{\sf good}^{\sf max}$, can be computed recursively as follows:
    initially, $S_{\sf good}^{\sf max} = \{Q_i\}$; at any iteration, if there exists a set $T\in S_{\sf good}^{\sf max}$ s.t. there exists a pair $(q,q') \in T \cap S_{\sf bad}$, we do the following operation:
    $S_{\sf good}^{\sf max} := (S_{\sf good}^{\sf max} \setminus T) \cup (T\setminus \{q\}) \cup (T\setminus \{q'\})$,
    with the condition that  $T \setminus \{q\}$ is added to the set $S_{\sf good}^{\sf max}$ only if there exists no $T'$ in $S_{\sf good}^{\sf max}$ such that $T' \supseteq T\setminus \{q\}$, and similarly for $T \setminus \{q'\}$. Then
    we can prove the correctness of the above procedure in the following manner.

\begin{theorem}
\label{th:heuristic-correctness}
    The above procedure terminates. Moreover, the set $S_{\sf good}^{\sf max}$ returned by the algorithm is the set of all maximal compatible sets in the 3DFA $D_i$.
\end{theorem}

\begin{proof}
 {\bf Termination:}
    The number of distinct subsets of $Q_i$ is bounded by $2^{|Q_i|}$, where $|Q_i|$ denotes the size of $Q_i$. Notice that, at every iteration, a set $T$ from $S_{\sf good}^{\sf max}$ is chosen and
    (at most two) sets with `strictly' smaller size gets added to $S_{\sf good}^{\sf max}$. Therefore, it is not possible for the same set to be added twice along the procedure. In the worst case, it can contain all the distinct subsets of $Q_i$. Since the number of such sets is bounded and no set is added twice, the procedure terminates.
    
 {\bf Correctness:}
 Let us denote by $S_0$ the set $\{Q_i\}$, and let $S_1, S_2, \ldots$ denote the sets obtained after each iteration of the algorithm. 
 Let us denote by $S_m$ the set returned by the algorithm, \emph{i.e.}, $S_m = S_{m+1}$. The proof of correctness follows from the following lemma:
 \begin{lemma}
 \label{lem:all-are-compat}
 \label{compat-in-good}
 All sets in $S_m$ are compatible, and are maximal w.r.t the subset ordering, \ie, they form a $\subseteq$-antichain. 
Furthermore, for every $S \subseteq Q_i$ that is compatible, there is some $S' \in S_m$ such that $S \subseteq S'$. 
    \end{lemma}
     \begin{proof}
   \emph{First}, if $S_m$ contains a set $T$ which is not compatible, then $T$ would be deleted at the $m+1$-th iteration of the algorithm, which contradicts the fact that $S_m= S_{m+1}$.
   Therefore, all sets in $S_m$ are compatible. 
    
   {\em The second statement} can be shown by induction on $0 \le j \le m$. 
   
   \noindent{\bf Base case:} $S_0 = \{Q_i\}$ is maximal. 
   
   \noindent{\bf Induction hypothesis:} Suppose all sets in $S_{j-1}$ is maximal. 
   
   \noindent{\bf Induction step:} Suppose we delete the set $T$ at iteration $j$, and let $(q,q')$ be the incompatible pair in $T$. Then $T\setminus \{q\} \in S_j \setminus S_{j-1}$ iff there is no  $T'\in S_{j-1}$ s.t. $T' \supseteq T\setminus \{q\}$ 
 iff $T\setminus \{q\}$ is maximal (same argument for $T \setminus \{q'\})$). Since these are the only two sets that could possibly be added to $S_j$, an since by induction hypothesis, all sets in $S_{j-1}$ are maximal, the statement follows also for $j$.
  
{\em The third statement} can also be proved by induction on $0 \le j \le m$.

\noindent{\bf Base case:} Since $S_0 = \{Q_i\}$, and since $S \subseteq Q_i$, the statement trivially holds.

\noindent{\bf Induction hypothesis:} Suppose the statement holds for $j-1$.

\noindent{\bf Induction step:} Let $S \subseteq Q_i$ be a compatible set. Let $T$ be the set selected in iteration $j$ with the incompatible pair being $(q,q')$. From induction hypothesis,  there exists a set $S' \in S_{j-1}$ such that $S \subseteq S'$. 

{\tt Case 1.} 
$S' \neq T$. In this case, $S' \in S_j$ as well, and hence the result follows.

{\tt Case 2.} $S'$ is $T$. Now, note that, in place of $T$, the set $S_j$ contains two sets $T_1, T_2$ (both are compatible, due to Lemma~\ref{lem:all-are-compat}) where $T_1 \supseteq (T \setminus \{q\})$ and $T_2 \supseteq (T \setminus \{q'\})$, and $(q, q')$ is an incompatible pair. Now, since $S$ is consistent, $S$ (and also, $T_1$, $T_2$) cannot contain both the states $q, q'$. Therefore, $S \subseteq T_1$ if $q \notin S$, and $S \subseteq T_2$ if $q' \notin S$.
\end{proof}\end{proof}
    \noindent{\bf In the final step}, we define a DFA $C_i = (Q_i',q'_{\sf init},(\sevents),\delta_i',F_i)$, where:
        $Q_i' = S_{\sf good}^{\sf max}$;
        $q'_{\sf init} = T \in S_{\sf good}^{\sf max}$ s.t. $q_{\sf init} \in T$, and $T$ has the highest cardinality;
        for any $T \in Q_i'$, and $e \in (\sevents)$, define $\delta_i'(T, e) = T'\in S_{\sf good}^{\sf max}$ s.t. $\bigcup_{q\in T} \{\delta_i(q,e)\} \subseteq T'$, and has the largest cardinality; and
        $F_i := \{T \in S_{\sf good}^{\sf max} \mid T \cap {\sf A_i} \neq \emptyset\}$.
        
    Note that, the DFA $C_i$ constructed from the 3DFA $D_i$ is not unique in the sense that (1) there might be more than one possible $T$ in $S_{\sf good}^{\sf max}$ s.t. $q_{\sf init} \in T$, and (2) there might be more than one possible $T'$ in $S_{\sf good}^{\sf max}$ such that $T' \supseteq \bigcup_{q\in T} \{\delta_i(q,e)\}$.
    Also, it is important here to note that, one can apply further heuristics at this step, in particular, while adding a new transition from $T$ on $e$, if there is already some maximal set $T'$ that has been discovered previously, set the target of this transition to $T'$, even if it is not of maximum cardinality.
\begin{lemma}
\label{C-consistent-with-D}
    For any 3DFA $D_i$, the procedure described above is well-defined. Then for any $C_i$ constructed according to the procedure, $C_i$ is consistent with $D_i$, \ie, $L(D^+_i) \subseteq L(C_i)$ and $L(D^-_i) \subseteq \overline{L(C_i)}$.
\end{lemma}

\begin{proof}
    {\bf Well-definedness:}
    \emph{First}, notice that, from Lemma~\ref{compat-in-good}, there must exist a set $T \in S_{\sf good}^{\sf max}$ such that $\{q_{\sf init}\} \subseteq T$.
    \emph{Second}, let $T$ be any maximal set in $S_{\sf good}^{\sf max}$, and  $e \in (\sevents)$. Let $\delta_i'(T, e) = T''$. Notice that, since $T$ is compatible, therefore $T''$ has to be a compatible set. 
    Indeed, if there exists a pair of states $(q, q')$ in $T''$ that are incompatible, then there must exist $q_1, q_1'$ in $T$ such that $q_1 \xrightarrow{e} q$ and $q_1' \xrightarrow{e} q'$ are two transitions in $E_i$, and hence $(q_1,q_1')$ are also incompatible contradicting the hypothesis that $T$ is compatible. 
    Now that we have shown that $T''$ must be a compatible set, again by Lemma~\ref{compat-in-good}, there must exist $T' \in S_{\sf good}^{\sf max}$ such that $T'' \subseteq T'$. This concludes the proof.

    {\bf Consistency:} Consider any $C_i$ constructed according to the procedure from the 3DFA $D_i$. 
    Suppose $\word = e_1 e_2\dots e_m$ is a region word over $(\Sigma, K)$ (not necessarily consistent).
    Let $\rho = q_{\sf init} = q_0 \xrightarrow{e_1} q_1 \xrightarrow{e_2} \ldots \xrightarrow{e_m} q_m$ be the unique run of $\word$ in the 3DFA $D_i$.
    By applying Lemma~\ref{compat-in-good} on the transitions of $\rho$, in particular by induction, one can show that there exists a unique run $\rho' = T_0 \xrightarrow{e_1} T_1 \xrightarrow{e_2} \ldots \xrightarrow{e_m} T_m$, where for all $0 \le j \le m$, $T_j$ is a state in $Q_i'$ s.t. $q_j \in T_j$.
    Now suppose, $\word \in L(D_i^+)$. Then, $q_m \in {\sf A_m} \cap T_m$. By definition of $F_i$ (the set of accepting states of $C_i$), $T_m \in F_i$, which implies that $\word \in L(C_i)$.
    A similar argument can be used to show that for any $\word \in L(D_i^-)$, $\word \in \overline{L(C_i)}$. This concludes the proof.
\end{proof}

Taking intersection with $\regL$ on both sides, together with Lemma~\ref{lem:complement-rw}, we get the following corollary:
\begin{corollary}
\label{cor:consistent}
    $\lrw{D_i^+} \subseteq \lrw{C_i}$, and
    $\lrw{D_i^-} \subseteq \overline{\lrw{C_i}} \cap \regL$.
\end{corollary}
\subsection{Checking if $C_i$ is a solution}
\label{sec:cex-processing}
Now that $C_i$ has been extracted from $D_i$, we can query the ${\tt Teacher}$ with two inclusion queries that check if 
$L^{\sf tw}(C_i) = L$. 
If the answer is yes, then $C_i$ is a solution to our the learning problem, and we return the automaton $C_i$; otherwise, the {\tt Teacher} returns a counterexample -- a consistent region word that is either accepted by $C_i$ but is not in $\RW{L}$, or a consistent region word that is rejected by $C_i$ but is in $\RW{L}$. 
Notice that the equality check has been defined above as an equality check between two timed languages, whereas, the counterexamples returned by the {\tt Teacher} are region words. This is not contradictory, indeed, one can show a similar result to Lemma~\ref{lem:rw-iff-tw} that $\ltw{C_i} = L$ iff $\lrw{C_i} = \RW{L}$.
More details on how the counter-example is extracted from the answer received from the {\tt Teacher} is described in Section~\ref{section:implementation}.
We will show in Lemma~\ref{lem:cex-vs-sound} that a counterexample to the equality check is also a `discrepancy' in the 3DFA $D_i$.
Therefore, this counterexample is used to update the observation table as in $L^*$  to produce the  
subsequent hypothesis $D_{i+1}$. 

\subsection{Correctness of the \tlsep\ algorithm}

  \begin{definition}[Soundness]
      \label{def:soundness}
      A 3DFA $D$ over $\sevents$ is {\em sound} w.r.t. the timed language $L$ if 
           the timed language of $D^+$, when interpreted as a DERA, contains $L$,
       \ie, $L \subseteq \ltw{D^+}$; and
       the timed language of $D^-$, when interpreted as a DERA, contains $\overline{L}$, \ie,
        $\overline{L} \subseteq\ltw{D^-}$.
  \end{definition}

Similar to Lemma~\ref{lem:complete-for-region}, one can show the following result:
\begin{lemma}
\label{lem:sound-for-region}
Let $L$ be a timed language over $\Sigma$ recognized by a $K$-DERA $A$ over $\Sigma$, then for any 3DFA $D$ over $\sevents$,
    $D$ is sound w.r.t. $L$ iff $\RW{L} = \lrw{A} \subseteq \lrw{D^+}$ and $\RW{\overline{L}} = \lrw{\overline{A}} \subseteq \lrw{D^-}$.
\end{lemma}
The following lemma shows that a counter-example of the equality check in Section~\ref{sec:cex-processing} is also a counter-example to the soundness of $D_i$: 
\begin{lemma}
\label{lem:cex-vs-sound}
    If a counter-example is generated for the equality check against $C_i$, then the corresponding 3DFA $D_i$ is not sound w.r.t. $\RW{L}$.
\end{lemma}
\begin{proof}
    Recall from Corollary~\ref{cor:consistent} that, 
    $L^{\sf rw}(D^+_i) \subseteq L^{\sf rw}(C_i)$, and 
    $L^{\sf rw}(D^-_i) \subseteq \overline{\lrw{C_i}}\cap \regL$. 
    Let $\word$ be a counter-example for the equality check, then:

\noindent (1) either there exists a consistent region word $\word \in \RW{L} \setminus L^{\sf rw}(C_i)$, which implies $\word \in \RW{L} \setminus L^{\sf rw}(D_i^+)$. Hence, $D_i$ is not sound w.r.t $\RW{L}$.

\noindent (2) or, there exists a consistent region word $\word \in  L^{\sf rw}(C_i) \setminus \RW{L}$. Since $\word$ is \emph{consistent}, using Lemma~\ref{lem:timed-language-union-of-regions}, we can deduce that $\word \in  \RW{\overline{L}} \setminus (\overline{\lrw{C_i}}\cap \regL)$, which implies $\word \in \RW{\overline{L}}\setminus L^{\sf rw}(D_i^-)$. Hence, $D_i$ is not sound w.r.t. $\lrw{A}$.
\end{proof}

We can prove the following using similar lines of arguments as in~\cite{CFCTW09}.
\begin{theorem}
\label{th:tlsep-correctness}
Let $L$ be a timed language over $\Sigma$ recognizable by a ERA given as input to the \tlsep\ algorithm.  Then,
\begin{enumerate}
    \item 
    the algorithm \tlsep\ terminates,
    \item 
    the algorithm \tlsep\ is correct,
    \ie, letting $C$ be the automaton returned by the algorithm,
    $L^{\sf tw}(C) = L$,
    \item  
    at every iteration of \tlsep, if strong-completeness is checked for 3DFA $D_i$, and if a minimization procedure is used to construct $C_i$, then $C$ has a minimal number of states.
\end{enumerate}
\end{theorem}
\begin{proof}
    1. One can define the following `canonical' 3DFA $D$ over $\sevents$ from $L$ as follows:
    for every \emph{consistent} region word $w$ s.t. $w \in \RW{L}$, define $D(w) = 1$; 
    for every \emph{consistent} region word $w$ s.t. $w \notin \RW{L}$, define $D(w) = 0$; and
    for every \emph{inconsistent} region word $w$, define $D(w) = {?}$. 
    It is then easy to verify that $D$ is sound and complete w.r.t. $L$.
    Due to Lemma~\ref{lem:cex-vs-sound}, at iteration $i$, a counterexample is produced only if the 3DFA $D_i$ is not sound. 
    Notice also that, from Lemma~\ref{lem:regL-is-reg} and Lemma~\ref{lem:srw-to-rw}, the languages $\regL$, $\overline{\regL}$, $\RW{L}$ and $\overline{\RW{L}}$ are regular languages.
    Therefore, the 3DFAs constructed from the observation table at every iteration of \tlsep\ algorithm, as described in Section~\ref{sec:table-to-3DFA}, gradually converges to the canonical sound and complete 3DFA $D$ (as shown in the case of regular separability in~\cite{GJL10}).
    
    2. Since $C$ is the final automaton, it must be the case that when an equality check was asked to the {\sf teacher} with hypothesis $C$, no counter-example was generated, and hence $\ltw{C} = L$.

    3. This statement follows from Lemma~\ref{lem:minimalility-3DFA} together with the fact that $C_i$ is the minimal consistent DFA for $D_i$.
    
\end{proof}

\section{Implementation and its performance}
\label{section:implementation}

We have made a prototype implementation of \tlsep\ in {\sc Python}. Parts of our implementation are inspired from the implementation of $L^*$ in AALpy~\cite{aalpy}. For the experiments, we assume that the {\tt Teacher} has access to a DERA $A$ recognizing the target timed language $L$. 
Below, we describe how we implement the sub-procedures.
Note that, in our implementation, we check for completeness, and not for strong-completeness, of the 3DFA's. 

\subparagraph*{Emptiness of region words.} Since we do not perform a membership query for inconsistent region words, before every membership query, we first check if the region word is consistent or not (in agreement with our greybox learning framework described in Page~\pageref{paragraph:learning-protocol}). The consistency check is performed by encoding the region word as an SMT formula in Linear Real Arithmetic and then checking (using the SMT solver Z3~\cite{Z3}) if the formula is satisfiable or not. 
    Unsatisfiability of the formula implies the region word is inconsistent. 
    
    \subparagraph*{Inclusion of Timed languages.}
    We check inclusion between timed languages recognizable by ERA in two situations: (i) while making the 3DFA $D_i$ complete (cf. Definition~\ref{def:completeness}), and (ii) during checking if the constructed DERA $C_i$ recognizes $L$ (cf. Definition~\ref{def:soundness}). Both of these checks can be reduced to checking emptiness of appropriate automata. 
    We perform the emptiness checks using the reachability algorithm ({\sf covreach}) implemented inside the tool {\sf TChecker}~\cite{Tchecker}.

    \subparagraph*{Counterexample processing.} When one of the language inclusions during completeness checking returns {\sf False},  
    we obtain a `concrete' path from \textsf{TChecker} that acts as a certificate for non-emptiness. The path is of the form:
    $${\sf cex} := (q_0, v_0) \xrightarrow[\{x_{\sigma_1}\}]{\sigma_1,~g_1} (q_1, v_1) \xrightarrow[\{x_{\sigma_2}\}]{\sigma_2,~g_2} (q_2, v_2) \ldots (q_{n-1}, v_{n-1}) \xrightarrow[\{x_{\sigma_n}\}]{\sigma_n,~g_n} (q_n, v_n)$$
    Since the guards present in the automata $D_i^+, D_i^-$ are region words, every guard present in the product automaton is also a region word. 
    From ${\sf cex}$ we then construct the region word $(\sigma_1, g_1)(\sigma_2, g_2) \ldots (\sigma_n, g_n)$. We then use the algorithm proposed by Rivest-Schapire~\cite{IS14,RS93} to compute the `witnessing suffix' ${\sf ws}$ from ${\sf cex}$ and then add ${\sf ws}$ to the set of suffixes ($E$) of the observation table.
    On the other hand, when we receive a counterexample during the soundness check, we add all the prefixes of this counterexample to the set $S$ instead. Note that, the guards present on the transitions in our output are simple constraints, and not constraints.

\begin{figure}[t]
    \centering
    \begin{subfigure}{.4\textwidth}
    \begin{tikzpicture}[scale = 0.7]
        \everymath{\scriptstyle}
        \begin{scope}[every node/.style={circle, draw, inner sep=2pt,
            minimum size = 3mm, outer sep=3pt, thick}] 
            \node [double] (0) at (0,0) {$q_0$}; 
            \node [] (1) at (2.1,0) {$q_1$}; 
            \node [double] (2) at (5.2,0) {$q_2$}; 
        \end{scope}
        \begin{scope}[->, thick]
            \draw (-0.7,0) to (0); 
            \draw [rounded corners] (0) to (-0.3,1) to (0.3, 1) to (0);
            \draw (0) to (1);
            \draw [rounded corners] (1) to (1.8,1) to (2.4, 1) to (1);
            \draw [bend left] (1) to (2);
            \draw (1) to (2);
            \draw [bend left] (2) to (1);
        \end{scope}
        \node at (0,1.2) {$a,~x_a > 0$};
        \node at (1.1,0.2) {$a,~x_a = 0$};
        \node at (2,1.2) {$a,~0 < x_a < 1$};
        \node at (3.5,0.2) {$a,~x_a = 0$};
        \node at (3.5,0.8) {$a,~x_a \geq 1$};
        \node at (3.5,-0.8) {$a,~x_a \geq 0$};
    \end{tikzpicture}
    \caption{Example-4}
    \label{fig:example-5}
    \end{subfigure}
    \begin{subfigure}{.5\textwidth}
    \begin{tikzpicture}[scale=0.7]
    \everymath{\scriptstyle}
        \begin{scope}[every node/.style={circle, draw, inner sep=2pt,
            minimum size = 3mm, outer sep=3pt, thick}] 
            \node [double] (0) at (-1,0) {$q_0$}; 
            \node [] (1) at (1.8,0) {$q_1$}; 
            \node [double] (2) at (6,0) {$q_2$}; 
        \end{scope}
        \begin{scope}[->, thick]
            \draw (-1.7,0) to (0); 
            \draw [rounded corners] (0) to (-1.4,1) to (-0.6, 1) to (0);
            \draw (0) to (1);
            \draw [rounded corners] (1) to (1.5,1) to (2.1, 1) to (1);
            \draw [bend left] (1) to (2);
            \draw (1) to (2);
            \draw [bend left] (2) to (1);
            \draw [rounded corners] (2) to (6,-2) to (1.8,-2) to (1);
        \end{scope}
        \node at (-1,1.2) { $a,~x_0 > 0,~\{x_0\}$};
        \node at (0.4,0.2) {$a,~x_0 = 0$};
        \node at (0.5,-0.2) {$\{x_1\}$};
        \node at (2,1.5) {$a,~0 < x_0 < 1 \wedge 0 < x_1 < 1$};
        \node at (2,1.2) {$\{x_0, x_1\}$};
        \node at (4,0.2) {$a,~x_0 \geq 1\wedge x_1 \geq 1$};
        \node at (4,-0.2) {$\{x_0,x_1,x_2\}$};
        \node at (4.1,1.1) {$a,~x_0=0 \wedge x_1=0$};
        \node at (4.95,0.8) {$\{x_2\}$};
        \node at (4,-0.9) {$a,~x_0 = 0 \wedge x_1 = 0 \wedge x_2 = 0$};
        \node at (4,-1.2) {$\{x_0, x_1\}$};
        \node at (4,-1.8) {$a,~x_0 > 0 \wedge x_1 > 0 \wedge x_2 > 0$};
        \node at (4,-2.2) {$\{x_0, x_1, x_2\}$};
    \end{tikzpicture}
    \caption{DTA computed by \textsf{LeanTA} on Figure~\ref{fig:example-5}}
    \label{fig:example-5-learnta}
    \end{subfigure}
    \caption{Explainability of \tlsep}
    \label{fig:ex-5-both}
\end{figure}

We now compare the performance of our prototype implementation of \tlsep\ against the algorithm proposed by Waga in~\cite{Waga23} on a set of synthetic timed languages and on a benchmark taken from their tool \textsf{LearnTA}. 
We witness encouraging results, among those, we stress on the important ones below.

\subparagraph*{Minimality.} We have shown in Section~\ref{sec:tlsep} that if one ensures strong-completeness of the 3DFAs throughout the procedure, then the minimality in the number of states is assured. However, this is often the case even with only checking completeness. For instance, recall our running example from Figure~\ref{fig:dera}. For this language, our algorithm was able to find a DERA with only two states (depicted in Figure~\ref{fig:tlsep-running-output}) -- which is indeed the minimum number of states. 
On the other hand, \textsf{LearnTA} does not claim minimality in the computed automaton, and indeed for this language, they end up constructing the automaton in Figure~\ref{fig:learnta-running-output} instead.
Note that for the sake of clarity, we draw a single edge from $q_0$ to $q_1$ on $a$ in Figure~\ref{fig:tlsep-running-output}, however, our implementation constructs  several parallel edges (on $a$) between these two states, one for each constraint in ${\sf Reg}(\Sigma,K)$.
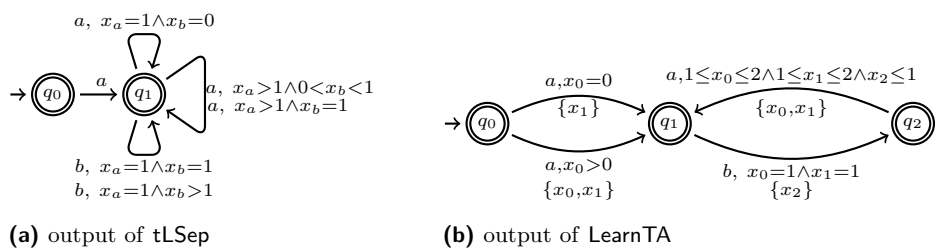
\begin{figure}[t]
    \centering
    \begin{subfigure}{.4\textwidth}
    \begin{tikzpicture}[scale=0.8]
    \everymath{\scriptstyle}
        \begin{scope}[every node/.style={circle, draw, inner sep=2pt,
            minimum size = 3mm, outer sep=3pt, thick}] 
            \node [double] (0) at (1.5,0) {$q_0$}; 
            \node [double] (1) at (3,0) {$q_1$};
        \end{scope}
        \begin{scope}[->, thick]
            \draw (0.8,0) to (0); 
            \draw (0) to (1);
            \draw [rounded corners] (1) to (2.7,1) to (3.3,1) to (1);
            \draw [rounded corners] (1) to (2.7,-1) to (3.3,-1) to (1);
            \draw [rounded corners] (1) to (4,0.7) to (4,-0.7) to (1);
        \end{scope}
        \node at (2.3,0.2) { $a$};    
        \node at (3,1.2) { $a,~x_a = 1\land x_b=0$};
        \node at (5.4,0.1) { $a,~x_a>1 \wedge 0 < x_b < 1$};
        \node at (5.2,-0.2) { $a,~x_a>1 \wedge x_b = 1$};
        \node at (3,-1.2) { $b,~x_a=1 \wedge x_b = 1$};
        \node at (3,-1.6) { $b,~x_a=1 \wedge x_b > 1$};
    \end{tikzpicture}
    \caption{output of \tlsep\ }
    \label{fig:tlsep-running-output}
    \end{subfigure}
    \begin{subfigure}{.55\textwidth}
        \begin{tikzpicture}[scale=0.8]
    \everymath{\scriptstyle}
        \begin{scope}[every node/.style={circle, draw, inner sep=2pt,
            minimum size = 3mm, outer sep=3pt, thick}] 
            \node [double] (0) at (0,0) {$q_0$}; 
            \node [double] (1) at (3,0) {$q_1$};
            \node [double] (2) at (7,0) {$q_2$};
        \end{scope}
        \begin{scope}[->, thick]
            \draw (-0.7,0) to (0); 
            \draw [bend right = 25] (0) to (1);
            \draw [bend left = 25] (0) to (1);
            \draw [bend right = 25](1) to (2);
            \draw [bend right = 25] (2) to (1);
        \end{scope}
        \node at (1.5,0.7) { $a,x_0=0$};
        \node at (1.5,0.25) { $\{x_1\}$};
        \node at (1.5,-0.7) { $a,x_0>0$};
        \node at (1.5,-1.1) { $\{x_0,x_1\}$};
        
        \node at (5,-0.8) { $b,~x_0 = 1\land x_1=1$};
        \node at (5,-1.1) { $\{x_2\}$};
        \node at (5,0.8) { $a,1\le x_0\le 2\land 1\le x_1\le 2\land x_2\le1$};
        \node at (5,0.25) { $\{x_0,x_1\}$};
    \end{tikzpicture}
    \caption{output of \textsf{LearnTA}}
    \label{fig:learnta-running-output}
    \end{subfigure}
    \caption{The automata computed by the two algorithms on Figure~\ref{fig:dera}}
\end{figure}

\subparagraph*{Explainability.} From a practical point of view, explainability of the output automaton is a very important property of learning algorithms. One would ideally like the learning algorithm to return a `readable', or easily explainable, model for the target language. 
We find some of the automata returned by \tlsep\ are relatively easier to understand than the model returned by {\sf LearnTA}. This is essentially due to the readability of ERA over TA. A comparative example can be found in Figure~\ref{fig:ex-5-both}.

\subparagraph*{Efficiency.} 
We treat completeness checks for a 3DFA and language-equivalence checks for a DERA, as equivalence queries ({\sf EQ}). Note that, {\sf EQ}'s are not new queries, these are merely (at most) two Inclusion Queries ({\sf IQ}).
We report on the number of queries for \tlsep\ and \textsf{LearnTA} in Table~\ref{tab:experiments} on a set of examples. 
For both the algorithms, we only report on the number of queries `with memory' in the sense that we fetch the answer to a query from the memory if it was already computed. 
Notice that \textsf{LearnTA} uses two types of membership queries, called \emph{symbolic queries} ({\sf s-MQ}) and \emph{timed-queries} ({\sf t-MQ}), and in most of the examples, \tlsep\ needs much lesser number of membership queries than \textsf{LearnTA}.

Although the class of  ERA-recognizable languages is a strict subclass of DTA-recognizable languages, we found a set of benchmarks from the \textsf{LearnTA} tool, that are in fact event-recording automata (after renaming the clocks). 
This example is called \emph{Unbalanced},  
Table~\ref{tab:experiments} shows that \tlsep\ uses significantly lesser number of queries than {\sf LearnTA} for these, illustrating the gains of our greybox approach.

\begin{figure}[t]
    \centering
    \begin{subfigure}{0.3\textwidth}
    \begin{tikzpicture}[scale=0.7]
    \everymath{\scriptstyle}
        \begin{scope}[every node/.style={circle, draw, inner sep=2pt,
            minimum size = 3mm, outer sep=3pt, thick}] 
            \node [double] (0) at (0,0) {$q_0$}; 
            \node [] (1) at (2,0) {$q_1$};
            \node [] (2) at (2,-2) {$q_2$};
            \node [] (3) at (0,-2) {$q_3$};
        \end{scope}
        \begin{scope}[->, thick]
            \draw (-0.7,0) to (0); 
            \draw (0) to (1);
            \draw (1) to (2);
            \draw (2) to (3);
            \draw (3) to (0);
        \end{scope}
        \node at (1,0.2) { $a$};
        \node at (2.2,-1) { $b$};
        \node at (1,-1.8) { $a,~x_a < 1$};
        \node at (-0.8,-1) { $b,~x_a > 1$};
    \end{tikzpicture}
    \caption{Example-1}
    \label{fig:example-1}
    \end{subfigure}
\begin{subfigure}{0.2\textwidth}
    \begin{tikzpicture}[scale=0.7]
    \everymath{\scriptstyle}
        \begin{scope}[every node/.style={circle, draw, inner sep=2pt,
            minimum size = 3mm, outer sep=3pt, thick}] 
            \node [double] (0) at (0,0) {$q_0$}; 
            \node [double] (1) at (2,0) {$q_1$};
        \end{scope}
        \begin{scope}[->, thick]
            \draw (-0.7,0) to (0); 
            \draw [bend left] (0) to (1);
            \draw [bend left] (1) to (0);
        \end{scope}
        \node at (1,0.6) { $a,~x_a = 1$};
        \node at (1,-0.6) { $b$};
    \end{tikzpicture}
    \caption{Example-2}
    \label{fig:example-2}
\end{subfigure}
\begin{subfigure}{0.4\textwidth}
    \begin{tikzpicture}[scale = 0.7]
    \everymath{\scriptstyle}
        \begin{scope}[every node/.style={circle, draw, inner sep=2pt,
            minimum size = 3mm, outer sep=3pt, thick}] 
            \node [] (0) at (0,0) {$q_0$}; 
            \node [] (1) at (2,0) {$q_1$};
            \node [double] (2) at (4.2,0) {$q_2$};
        \end{scope}
        \begin{scope}[->, thick]
            \draw (-0.7,0) to (0); 
            \draw [] (0) to (1);
            \draw [rounded corners] (1) to (1.7,1) to (2.3,1) to (1);
            \draw [] (1) to (2);
        \end{scope}
        \node at (1,0.2) { $a$};
        \node at (2,1.2) { $b,~x_a = 1$};
        \node at (3.1,0.2) { $b,~x_a > 1$};
    \end{tikzpicture}
    \caption{Example-3}
    \label{fig:example-4}
\end{subfigure}
\caption{Automata corresponding to the models ex1, ex2 and ex3 of the table, respectively}
\label{fig:example123}
\end{figure}
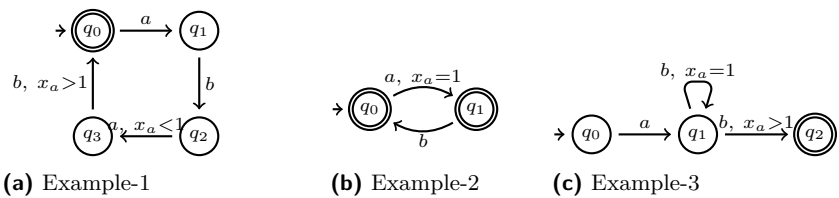

Below, we give brief descriptions of languages that we have learnt using our \tlsep\ algorithm. 
The timed language (ex1) has untimed language $(abab)^*$, where in every such four letter block, the second $a$ must occur before $1$ time unit of the first $a$ and the second $b$ must occur after $1$ time unit since the first $a$.
The language (ex2) has untimed language $(ab)^*$, where the first $a$ happens at $1$ time unit and every subsequent $a$ occurs exactly $1$ time unit after the preceding $a$. 
The timed language (ex3) 
 consists of timed words whose untimed language is $ab^*b$, and the first $a$ can occur at any time, then there can be several $b$'s all exactly $1$ time unit after $a$, and then the last $b$ occurs strictly after $1$ time unit since the first $a$. 
We also consider the language (ex4) represented by the automaton in Figure~\ref{fig:example-5}  (this language is taken from \cite{GJL10}), for whose language, as we have described, \tlsep\ constructs a relatively more understandable automaton compared to \textsf{LearnTA}.

\begin{table}[t]
    \centering
    \begin{tabular}{||c|c|c|c||c|c|c||c|c|c||}
\hline
\multirow{2}{*}{Model} & \multirow{2}{*}{$K$} & \multirow{2}{*}{$|Q|$} & \multirow{2}{*}{$|\Sigma|$} & \multicolumn{3}{c||}{\tlsep} &  \multicolumn{3}{c||}{\textsf{LearnTA}} \\
& & & & ${\sf MQ}$ & ${\sf IQ}$ &${\sf EQ}$ & ${\sf s-MQ}$& ${\sf t-MQ}$ & ${\sf EQ}$ \\
\hline
    \hline
    Figure~\ref{fig:dera} & 1 & 4 & 2 & 98 & 8 &5& 230&219& 4\\
    \hline
    ex1 & 1 & 5 & 2 & 219 & 11 & 6 & 296&365 &4  \\ 
    \hline 
    ex2 & 1 & 3 & 2 & 220 & 12 & 7& 247&233& 5\\ 
    \hline
    ex3 & 2 & 4 & 2 & 87 & 7 &4&123&173 &3 \\
    \hline
    ex4 & 1 & 3 & 1 & 26 & 5 &3& 40&43&2 \\
    \hline
    \hline
    Unbalanced-1 & 1 & 5 & 3 & 421 &17  &12&1717 &2394 &6\\
    \hline
    Unbalanced-2 & 2 & 5 & 3 & 1095 &27 &20 &7347& 13227 &10\\
    \hline
    Unbalanced-3 & 3 & 5 & 3 & 2087 & 37& 28&21200&45400 &17\\
    \hline
\end{tabular}
\caption{Experimental results. Here $K$ denotes the maximum constant appearing in the automaton, $|Q|$ and $|\Sigma|$ denote the numbers of states (of the automaton provided as input to the algorithm) and the alphabet. Note that, for some of these automata, the number of states specified in the table is one more than the number of states depicted in the figures. This is because, for the automata that are not total, we add an additional sink state and add all missing transitions to this state.}
\label{tab:experiments}
\end{table}

\section{Discussions}

Our grey box learning algorithm can produce automata with a minimal number of states, equivalent to the minimal DERA for the target language. Currently, the transitions are labeled with region constraints. In the future, we plan to implement a procedure that consolidates edges labeled with the same event and regions, whose unions are convex, into a single edge with a zone constraint. This planned procedure aims to not only maintain an automaton with the fewest possible states, but also to minimize the number of edges and zones. Such optimization could further enhance the readability of the models we produce, thereby improving explainability—an increasingly important aspect of machine learning.

Building on our proposed future work above, we also aim to refine our approach to membership queries. Instead of defaulting to queries on regions where it may not be necessary, we plan to consider queries on zones. If the response from the {\tt Teacher} is not uniform across the zone, only then will we break down the queries into more specific sub-queries, ending in regions if necessary. This strategy will necessitate revising the current learning protocol we use, but exploring these alternatives could enhance the efficiency and effectiveness of our learning process. Similar ideas have been used in~\cite{LADSL11}.

Throughout this paper, we follow the convention used in prior research~\cite{GJL10,Waga23} that the parameter $K$ is predetermined (the maximal constant against which the clocks are compared). Similarly, in our experimental section, we assume we know which events need to be tracked by a clock. 
While these assumptions are typically reasonable in practice, it is possible to consider a more flexible approach. We could start with no predetermined events to track with clocks and only introduce clocks for events and update the maximum constant, as needed. 
These could be inferred from the counterexamples provided by the {\tt Teacher}.

Overall, we believe the greybox learning framework described in this work may also be well-suited for other classes of languages that enjoy a similar condition like Equation~\ref{eq:intersection}, where the automaton $C$ can be of relatively smaller size, while a canonical model recognizing $L$ might be too large.

\bibliographystyle{plainurl}
\bibliography{main}

\end{document}